\documentclass{egpubl}
\usepackage{eg2027}

\ConferencePaper

\usepackage[T1]{fontenc}
\usepackage{dfadobe}

\usepackage{cite}  
\BibtexOrBiblatex
\electronicVersion
\PrintedOrElectronic
\ifpdf \usepackage[pdftex]{graphicx} \pdfcompresslevel=9
\else \usepackage[dvips]{graphicx} \fi

\usepackage{egweblnk}


\usepackage{amsmath,amssymb,amsthm}
\usepackage{booktabs,multirow,array}
\usepackage{xcolor}
\usepackage{enumitem}

\newcommand{\R}{\mathbb{R}}
\newcommand{\upre}{\lesssim}
\newcommand{\succap}{\gtrsim}
\newcommand{\norm}[1]{\left\lVert #1\right\rVert}
\newcommand{\abs}[1]{\left\lvert #1\right\rvert}
\newcommand{\one}{\mathbf{1}}
\newcommand{\dd}{\,\mathrm{d}}

\newcommand{\T}{\mathbb{T}}
\newcommand{\PS}{\mathcal{P}}
\newcommand{\Pix}{\mathrm{Pixel}}
\newcommand{\mus}{\mu_S}
\newcommand{\mut}{\mu_\theta}
\newcommand{\piv}{\pi_{v,\theta}}
\newcommand{\nuv}{\nu_{v,\theta}}
\newcommand{\Ccal}{\mathcal{C}}
\newcommand{\Vcal}{\mathcal{V}}
\newcommand{\Lcal}{\mathcal{L}}
\newcommand{\Rcal}{\mathcal{R}}
\newcommand{\gt}{\mathrm{gt}}
\newcommand{\occ}{\mathrm{occ}}

\newcommand{\proved}{\checkmark}

\newtheorem{theorem}{Theorem}
\newtheorem{lemma}[theorem]{Lemma}
\newtheorem{corollary}[theorem]{Corollary}
\newtheorem{definition}[theorem]{Definition}

\title[Budgeted-GS]%
      {Budgeted-GS: Real-Time Large-Scale Gaussian Splatting via
Factoring LOD}

\author[H. Wang]{Haipeng Wang\\ Neusoft\\ \texttt{haipeng\_wang@neusoft.com}}

\makeatletter
\def\p@copyrightTextLong{Preprint.}
\def\p@copyrightTextShort{Preprint.}
\def\p@copyrightTextShortEven{Preprint.}
\makeatother

\hypersetup{pdftitle={Budgeted-GS: Real-Time Large-Scale Gaussian Splatting via Factoring LOD},pdfauthor={Haipeng Wang},pdfsubject={arXiv preprint},pdfkeywords={3D Gaussian Splatting, level of detail, optimal transport, rate-distortion, capacity floor}}

\begin{document}

\maketitle

\begin{abstract}
3D Gaussian Splatting achieves excellent visual quality with real-time
rendering, but at the scale of entire cities it does not fit: a trained
model carries millions of primitives and gigabytes of memory, and
real-time rendering at high quality on a consumer GPU remains out of
reach. We introduce \emph{Budgeted-GS}, a post-hoc method that turns any
trained 3DGS model into a factoring tree, a multi-resolution
hierarchy of moment-matched aggregates. After a construction pass of a
few seconds, a single quality parameter then selects, for each view, the level of
detail that fits the memory of the target device, so the same
city-scale model serves GPUs with widely different memory capacities. When a new scene is to be
trained, the same theory applies: instead of growing a full-sized
model and compressing it afterwards, budget-centered training first
measures how many primitives the scene needs and then trains the
model directly at that size, avoiding the wasted effort of
optimizing primitives that are later discarded. Both methods are grounded in a measurable capacity
floor, a budget--error law derived from optimal transport in phase
space; selection rules certified by recent covering theorems decide
which primitives are redundant. The floor answers how many primitives a
scene actually needs and how many can safely be given up. We validate the floor on 13 public scenes under a
preregistered protocol, and exercise both methods from object scenes to
an official city capture, rendering it at native $1920\times1080$, full SH, in real
time on one consumer GPU.


\end{abstract}

\section{Introduction}\label{sec:intro}
Real-time radiance-field rendering of captured scenes is now routine at
object and room scale. 3D Gaussian Splatting (3DGS)~\cite{kerbl2023}
represents a scene as a set of anisotropic Gaussians rendered by projective
splatting with depth-sorted $\alpha$-compositing. At city scale it does not
fit: a trained model carries millions of primitives and gigabytes of
memory, and real-time rendering at high quality on a consumer GPU remains
out of reach. The obstacles are structural. At training time the standard
recipe is to densify a large model and then compress it; the primitives
are optimized for a different capacity, and the time spent on the removed
half is wasted. At deploy time the only degradation mode offered by
streaming renderers such as Gaussian Point Splatting
(GPS~\cite{rijsdijk2026gps}) is random subsampling. And beneath both sits
an unanswered question: nobody can say what \emph{the} right budget is.
The question ``how many primitives do I need here?'' is answered by
habit, because no principled budget--error floor has been available
either at training or deployment time.

This paper answers it with a rate--distortion theory of rendering
primitives: a trained model is a point on a \emph{rate--distortion
curve}; its primitive count is its budget, its error a function of that
budget. The theory has two mathematical halves. The \emph{capacity
floor}, derived from optimal transport in phase space (position
$\times$ direction, where rendering primitives and renderers are both
natural objects), says that if the scene's phase-space measure has $d_{\mathrm{ps}}$-dimensional
support, the minimal achievable error with $N$ primitives falls as
$N^{-2/d_{\mathrm{ps}}}$ (Zador's law). The \emph{collapse license} comes
from the 2025 resolution of the three-dimensional Kakeya set
conjecture~\cite{wz2025}: the Wang--Zahl union-volume estimate, in
weighted (continuous-opacity) form, certifies that wherever Gaussians
locally over-cover space, they merge into few representatives with
bounded volumetric loss. The floor tells a budget what it can achieve;
the license tells it where redundancy lives and what may be merged or
dropped.

The floor and the license are operationalized as two methods. On the
deployment side, the floor's allocation law ($N_j^*\propto K_j s_j/d_j$
per tile) is realized as a \emph{factoring tree} (Factoring LOD,
Section~\ref{sec:lod}), an octree of scale bands with moment-matched
equivalent planks, driven by a one-parameter pixel-footprint scheduler.
A single trained model sweeps a continuous memory--quality curve instead
of a single fixed operating point. On the training side,
\emph{Budget-Centered Training} (BCT, Section~\ref{sec:bct}) trains the
densifier directly at the floor's knee, so the model is born at the
right size instead of growing huge and being pruned afterwards. Beyond
rendering, the floor answers how many primitives a scene needs, how
much of it is redundancy, and how a fixed budget should split across
a city. Measuring it on 13 public scenes under a preregistered
protocol shows published compressors lying on an essentially flat
frontier. Their reported gains come from redundancy, not from
proximity to capacity (Section~\ref{sec:measure}). At city scale the same
discipline turns real-time. A retrain of the official city model serves
the official test trajectory at native $1920\times1080$, full SH, on
one consumer GPU---to our knowledge the first post-hoc, full-SH result
on a consumer GPU in this regime (Section~\ref{sec:exp-rt}).

\paragraph{Contributions.}
\begin{itemize}
  \item We introduce \emph{Factoring LOD}, a post-hoc layer that
  converts one trained model into a factoring tree with a one-parameter
  scheduler; it renders a city-scale scene at high quality in real time
  within whatever memory the consumer GPU has.
  \item We propose \emph{Budget-Centered Training}, which trains the
  densifier directly at the floor-located budget; the model is born at
  the right size instead of growing huge and being pruned afterwards;
  it leads the train-then-prune route at every matched budget, in a
  single training run.
  \item We establish the theory beneath the problem: a measurable
  capacity floor that answers, for any city-scale 3DGS scene, how many
  primitives it actually needs, and therefore how many can safely be
  given up.
\end{itemize}

\section{Related Work}\label{sec:related}

\textbf{Rendering large-scale scenes.} Every classic representation
solved scale with its own multiresolution answer. Points: surfel
splatting~\cite{pfister2000surfels}, QSplat's sphere
hierarchy~\cite{rusinkiewicz2000qsplat}, and Potree's out-of-core
octree generation~\cite{schutz2021octree} stream large point clouds
at interactive rates. Meshes: progressive
meshes~\cite{hoppe1996progressive} introduced view-dependent
refinement, and Nanite~\cite{karis2021nanite} virtualizes
micropolygon geometry in production engines. Voxels: sparse voxel
octrees~\cite{laine2010sparse} ray-march gigavoxel volumes. What these
systems share, a coarse-to-fine hierarchy carrying a per-frame
budget, is exactly what a raw 3DGS model~\cite{kerbl2023} lacks: it is
a flat set of primitives with no hierarchy at all.

\textbf{Data structures for large-scale 3DGS scenes.} Inside the 3DGS world, the
standard answer is a spatial structure coupled to training. Octree-GS
builds an octree of level-of-detail clusters~\cite{ren2024octreegs};
Scaffold-GS anchors Gaussians on a sparse grid~\cite{lu2023scaffold};
multiresolution hash grids organize NeRF-side
training~\cite{muller2022instant}; CityGaussian partitions the scene
into blocks with distance-based LOD~\cite{liu2024citygaussian},
following the block decomposition of
Block-NeRF~\cite{tancik2022blocknerf}; hierarchical 3DGS trains a
chunk pyramid~\cite{kerbl2024h3dgs}; concurrent efforts pursue
out-of-core streaming hierarchies~\cite{windisch2026alod} and LOD-aware training~\cite{kulhanek2025lodge}; GPS streams an octree
with distance shells~\cite{rijsdijk2026gps} (our cross-renderer
comparisons use it, Section~\ref{sec:lod}). These structures carry the
indexing, but the budget lives in the training pipeline: levels and
blocks are fixed when training ends. Our factoring tree is built
\emph{post-hoc} from any model, selects per (leaf, scale-band) instead
of per block or level, and grounds the selection in measurable
redundancy and pixel-footprint criteria (Section~\ref{sec:lod});
Section~\ref{sec:exp-e1} compares directly against Octree-GS and
CityGaussianV2's runtime LOD~\cite{liu2025citygaussianv2}.

\textbf{Quality and memory control for large-scale 3DGS scenes.} Once a model exists, a large
family controls its size: MCMC-based
densification~\cite{kheradmand2024mcmc} and regularization-based
pruning~\cite{mallick2024taming} budget during training;
LightGaussian~\cite{fan2024light},
Compact3DGS~\cite{lee2024compact3d},
RDO-Gaussian~\cite{wang2024rdogaussian},
EAGLES~\cite{girish2024eagles}, RadSplat~\cite{niemeyer2024radsplat},
PUP~3D-GS~\cite{hanson2025pup}, and RAP~\cite{yang2026rap} prune,
quantize, distill, or reselect primitives after training; 2DGS changes
the primitive itself to a surfel for surface
quality~\cite{huang2024surfel}. All of these act once, at training or
export time, and expose no runtime memory--quality trade-off or error
bound. Our budget parameter acts on any trained model continuously and
renderer-agnostically; LightGaussian's own score, applied post-hoc
through our export path, yields a Vanilla-flat budget axis to $25\%$ of
the primitives on city-scale scenes (Section~\ref{sec:exp-rt}); and
rasterizer accelerators~\cite{hanson2025speedy,ding2026sharpgs} are
composable with it: they render the selected set faster, while the
parameter decides what is selected.


\section{Overview}
\label{sec:overview}

Budgeted-GS applies one budget discipline at every stage of a
scene's life cycle
(Fig.~\ref{fig:pipeline}). The capacity floor
(Section~\ref{sec:theory}) gives, for any scene, the minimal
achievable rendering error at a primitive count $N$; it is stated
as theorems and measured per scene with a preregistered protocol
(Section~\ref{sec:measure}). Two
methods spend the budget where the floor says.
Budget-Centered Training (Section~\ref{sec:bct}) reads the
floor's knee $N^*$ off the measured curve and trains the densifier
directly at that budget. \emph{Factoring LOD}
(Section~\ref{sec:lod}) converts any trained model into a
factoring tree whose one-parameter pixel-footprint scheduler selects,
per view, exactly the working set that fits the memory at hand.
Section~\ref{sec:experiments} reports the experiments.
The supplementary material carries the
proofs.

\begin{figure*}[t]
\centering
\includegraphics[width=\linewidth]{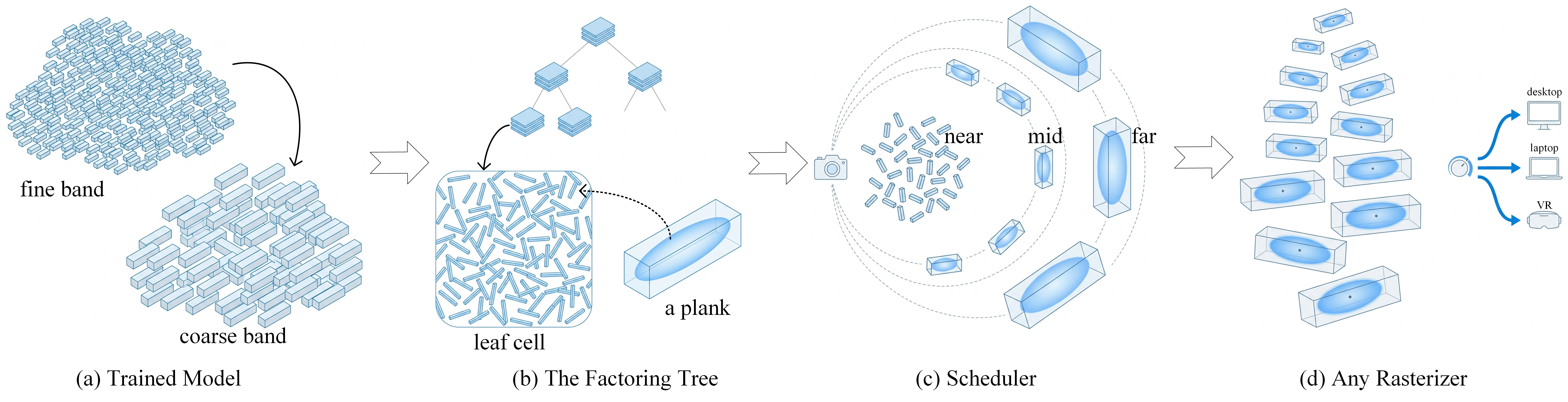}
\caption{Overview: the pipeline in four steps. (a) A trained model is
a flat set of raw planks in two scale bands. (b) Factoring
builds the leaf $\times$ scale-band tree
of moment-matched equivalent planks. (c) The one-parameter scheduler
selects, per view, raw / fine /
aggregate content by the pixel-footprint threshold $t=c\cdot d/f$.
(d) The emitted stream feeds any splatting rasterizer; the city-scale
payoff is Fig.~\ref{fig:cityscale}. Illustration drawn with AI
assistance and verified by the authors.}
\label{fig:pipeline}
\end{figure*}

\section{An optimal-transport theory for 3D Gaussian primitives}
\label{sec:theory}

Why should the budget be a design variable, and how should it be
spent? On real data, reconstruction error falls steeply as the budget
grows, then flattens into a shallow floor. This section develops a theory that predicts where that floor lies:
from primitives and rendering at the bottom, through construction and
capacity, up to the selection rules on top. Concretely, for a quality target $Q$, the theory yields (i) the minimal
primitive count a scene requires, (ii) the redundant subset of its
primitives, and (iii) the allocation of a fixed budget across a city.

\subsection{Rendering as optimal transport in phase space}

Prior work has employed optimal transport in 3DGS only in narrow,
task-specific roles, including Wasserstein losses between Gaussian
mixtures \cite{song2025d2gs}, radial-optimal-transport fitting of
structured Gaussians \cite{zhang2024gaussiancube}, and capacity
control in multi-view stylization \cite{wen2026capacity}. To our
knowledge, we are the first to formulate rendering as optimal
transport in phase space.

\begin{definition}[Phase space and primitives]
\label{def:render}
$\PS:=X\times\Omega$, the product of scene space $X\subset\R^3$
(the unit ball) and the direction sphere $\Omega=S^2$, is
\emph{phase space}; a cell $Q=V\times C$ (spatial
plank $\times$ directional cone) obeys $\Delta x\cdot\Delta\xi\ge c$ (the
uncertainty bound). A rendering primitive is a triple $(Q,w,\rho)$: cell,
opacity weight $w$, value model $\rho$. Three families exist. \textbf{Point} primitives take the limit
$\Delta x\to0$. \textbf{Wave} primitives are Gaussian wave packets with
$\Delta x\Delta\xi\asymp c$, the family of 3DGS surfels and 2DGS plates.
\textbf{Field} primitives are continuous densities, such as NeRF/SDF. A representation is a parameterized measure family
$\{\mut\}_{\theta\in\Theta}$ on $\PS$.
\end{definition}

\begin{definition}[Rendering coupling and marginalization]
For view $v$ and parameters $\theta$, the renderer defines a joint measure
$\piv$ on $\PS\times\Pix_v$, the pixel grid of view $v$, with source
marginal $\mut$ and pixel marginal
$\nuv$ (pixel masses $S(p):=\nuv(\{p\})$); the pixel colour is the conditional expectation
\begin{equation}
I_{v,\theta}(p)=\int_{\PS} \rho(x)\,\dd\piv(x\mid p).
\label{eq:ce}
\end{equation}
\end{definition}

\subsection{An optimal-transport view of 3DGS}

The framework rests on four axioms: every representation
induces a phase-space distribution $\mut$; the renderer is a legal
transport and the image is its coupling marginal; every construction
algorithm's objective is expressible as one transport cost; a single
source is transported to all view targets. For 3DGS-style compositing,
this coupling is the familiar compositing rule: $\piv(p,\{i\})\propto w_i(p)=\alpha_i(p)T_i(p)$
with $T_i=\prod_{j<i}(1-\alpha_j)$. Two problems follow. The first is \emph{rendering transport}: the renderer transports $\mut$ to per-view pixel marginals; the second is \emph{representation formation}: given complexity
$\mathrm{comp}_F(\theta)\le N$, minimize the transport cost
\begin{equation}
\Rcal_{\mathrm{res}}(N)=\min_{\theta}\sum_{v}\int_{\PS\times\Pix_v}
c\,\dd\piv,\quad c(x,p):=\norm{\rho(x)-c^{\gt}_v(p)}^2.
\label{eq:p8}
\end{equation}
Gradient training, sampling, meshing, and voxelization are solvers of
the second problem; their outputs are measured by the transport residual
(their convergence is proved or cited to classical sources in the supplementary material). The bottom layer is one theorem: the renderer
really is a legal transport, rendering is conditional expectation, and
the error splits into a tight three-term chain.

\begin{theorem}[The renderer, the rendering map, and the error chain]\label{thm:u123}
Under measurable ray--pixel assignment, bounded
transmittance-weighted mass, and Lipschitz regularity, every
$(v,\theta)$ admits a canonical coupling $\piv$ with marginals
$(\mut,\nuv)$, determined by the rendering rule and weakly continuous in $\theta$. With
projection Jacobian $J_\Pi$ and footprint kernel $G$, the pixel
probability kernel is
$\kappa_v(p,\dd x)\propto w(x)G(p-\Pi(x))\abs{\det J_\Pi(x)}\mathbf 1_V(x)$,
In the audited low-opacity domain, the pieces combine into a tight
three-term chain: the operator norm is sharp
($\norm{P_vf}_{L^2(\Pix_v)}\le\sqrt{\mu_v}\norm{f}_{L^2(\PS)}$), the
blending gap is second-order in the max opacity $\varepsilon$ (with
multiplicity $M(x)$:
$\abs{C(x)-L(x)}\le \frac{K'}{2}M(x)^2\varepsilon^2+O\bigl(K'M(x)\varepsilon^3\bigr)$), and
\begin{equation}
\norm{\nuv-\nu^{\gt}_v}_{L^2}
\le\sqrt{\ell_v}\,\norm{\mut-\mus}_{L^2}
+\frac{K'}{2}\varepsilon^2\norm{M^2}_{L^2}+\varepsilon_{\mathrm{target}}(v),
\end{equation}
a tight three-term chain: source gap, blending gap, target map; here
$\mu_v=\sum_pS(p)$ is the transported mass of view $v$, $\ell_v$ the
Lipschitz constant of the target map, $K'$ the compositing Lipschitz
constant, and measure norms are taken over densities with respect to
the phase-space reference measure, the source one over the normalized
source marginal of the coupling.
\end{theorem}

Forming a representation under a budget is itself an optimal-transport
problem~\cite{kantorovich1942,villani2009}, and its duality is what
makes the floor certifiable.

\begin{theorem}[Representation formation as restricted optimal transport]\label{thm:u4}
For compositing with normalized coupling $\tilde\pi_p$, the loss
decomposes as $\Lcal(\theta)=\Ccal(\theta)-\Vcal(\theta)$: the
transport cost $\Ccal$ minus the conditional variance $\Vcal$.
For atomic families with quadratic cost,
$\inf_{\abs{\operatorname{supp}\mut}\le N}\int\min_i\norm{x-y_i}^2\dd\mus
=W_2^2$-optimal quantization, and Lloyd's algorithm is monotone.
The certificate follows: the residual satisfies $\Rcal_{\mathrm{free}}(N)\le
\Rcal_{\mathrm{res}}(N)$ with gap decomposition
$\Delta(\theta)=\Delta_{\mathrm{coupl}}(\theta)+\Delta_{\mathrm{src}}(\theta)
\ge0$; for every dual-feasible $(\varphi,\psi)$ the certificate
\begin{equation}
\Lcal(\theta)\ge\int_{\PS}\varphi\,\dd\mut+\sum_p S(p)\psi(p)-\Vcal(\theta)
\label{eq:cert}
\end{equation}
lower-bounds achieved error; and with target colour measure $\nu_c$ of
intrinsic dimension $d_c\le3$,
$\Rcal_{\mathrm{free}}(N)\asymp C_{d_c}\norm{\nu_c}_{d_c/(d_c+2)}N^{-2/d_c}$.
\end{theorem}

\begin{figure*}[t]
\centering
\includegraphics[width=0.92\textwidth]{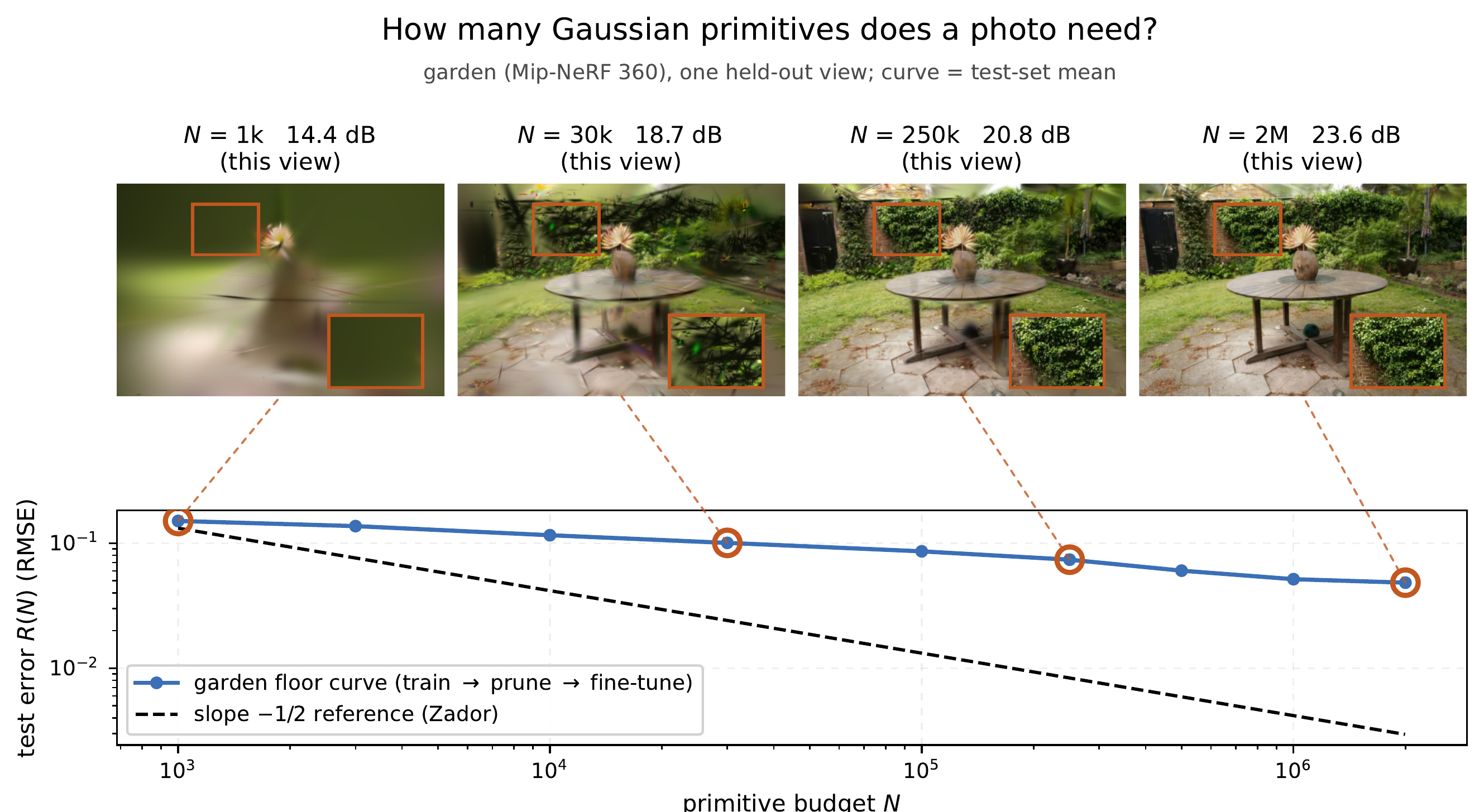}
\caption{The measured floor curve (garden, Mip-NeRF 360): budget ladder at four budgets against the $-1/2$ Zador reference.}
\label{fig:teaser}
\end{figure*}

\subsection{The capacity floor}\label{sec:floor}

How fast must this residual fall as the budget $N$ grows? Classical
high-resolution quantization theory~\cite{zador1982,graf2000foundations}
answers exactly this, and the rendering instance, whose source is the
phase-space measure of Definition~\ref{def:render}, obeys the same law:

\begin{theorem}[The capacity floor of the rendering instance]\label{thm:u51}
Imported from classical high-resolution quantization
\cite{zador1982,graf2000foundations} through the transport formulation
above, with the lower bound additionally assuming view completeness
(the view set separates the source): if the phase-space measure $\mus$ has $d_{\mathrm{ps}}$-dimensional
support with density $\lambda$, the minimal achievable rendering error
with $N$ primitives obeys
\begin{equation}
\Rcal(N)\ \asymp\ C_{d_{\mathrm{ps}}}\,\norm{\lambda}_{\frac{d_{\mathrm{ps}}}{d_{\mathrm{ps}}+2}}\,N^{-2/d_{\mathrm{ps}}}.
\label{eq:zador}
\end{equation}
On surface scenes the phase-space support is four-dimensional (surface
position $\times$ incident-direction hemisphere), so $\Rcal(N)\propto
N^{-1/2}$.
For the cartoon class with appearance smoothness $p$ and resolution $L$,
$\Rcal_F(N,L)^{1/2}\asymp C_1KN^{-1/2}(\log N)^c+C_2K'L^{-p}$ with optimal
allocation $N^*\asymp M^{p/(p+1)}$, $L^*\asymp M^{1/(2(p+1))}$;
$\mathrm{width}$ brackets sandwich $\Rcal_F^{\mathrm{free}}\le\Rcal_F\le
C\,\mathrm{width}$.
\end{theorem}
The free relaxation of Theorem~\ref{thm:u4} is target-limited (rate
$N^{-2/d_c}$), while this rendering instance is source-limited at
$N^{-2/d_{\mathrm{ps}}}$ with $d_{\mathrm{ps}}\ge d_c$: the requirement
that primitives must render, not merely match colours, is what slows
the rate. The floor is \emph{certifiable} rather than assumed: the duality chain
$L+V=\Ccal\ge\Ccal_{\mathrm{free}}\ge D$ evaluates, per budget
and per test view, a lower bound on the achievable error
(Eq.~\eqref{eq:cert}). It is also
\emph{measurable}, and Section~\ref{sec:measure} reports the curve for 13 public scenes.

Three supporting results, proved, cited, or labelled as sketches in the supplementary material, round
out the framework: gradient training is a legitimate solver of the
restricted-transport problem and attains the classical rates of
sampling, Lloyd, and meshing; cross-view colour agreement is
structural, and the new-view error is bounded by the training error
plus an angular-regularity term; and across primitive families on
surfaces, triangles achieve $N^{-1}$, point-like primitives
$N^{-1/2}$, and isotropic voxels $N^{-1/4}$. One part reads as a
prediction: the optimal cell is an anisotropic, covariance-ellipsoidal
atom, exactly what 3DGS fits, and exactly what the moment-matched
aggregates of the factoring tree (Section~\ref{sec:lod-tree})
construct.

\subsection{The selection rules}

The floor is one half of the mathematics; the other half is a set of
\emph{selection rules}, theorems that decide which primitives may be
merged or dropped at a given budget. Here a rendering primitive is
viewed geometrically, as a \emph{plank} $(\mu,R,h,\alpha)$: an oriented
box centered at $\mu$, rotated by $R$, of half-extents $h=3\sigma$,
carrying the opacity $\alpha$ as its weight. We state the rules here;
the proofs are
in the supplementary material.

\begin{definition}[Essentially distinct tubes and the convex Wolff axioms]
\label{def:wolff}
Sets $U,V\subset\R^3$ are \emph{essentially distinct} if
$\abs{U\cap V}\le\tfrac12\max(\abs{U},\abs{V})$; throughout, $\T$ is a
set of essentially distinct $\delta\times\delta\times1$ tubes in the
unit ball, so $\abs{T}\asymp\delta^2$. Following \cite{wz2025}
(Def.~1.3), $\T$ \emph{obeys the Katz--Tao convex Wolff axioms with
error $C$} if
\begin{equation}
\#\{T\in\T:\ T\subset W\}\;\le\;C\,\abs{W}\,\abs{T}^{-1}
\qquad\text{for every convex } W\subset\R^3,
\label{eq:ktcw}
\end{equation}
i.e.\ no convex set concentrates more tubes than its volume admits; a
Kakeya-type arrangement (one tube per $\delta$-separated direction)
satisfies \eqref{eq:ktcw} with $C\upre1$.
\end{definition}

\begin{figure*}[t]
\centering
\includegraphics[width=\linewidth]{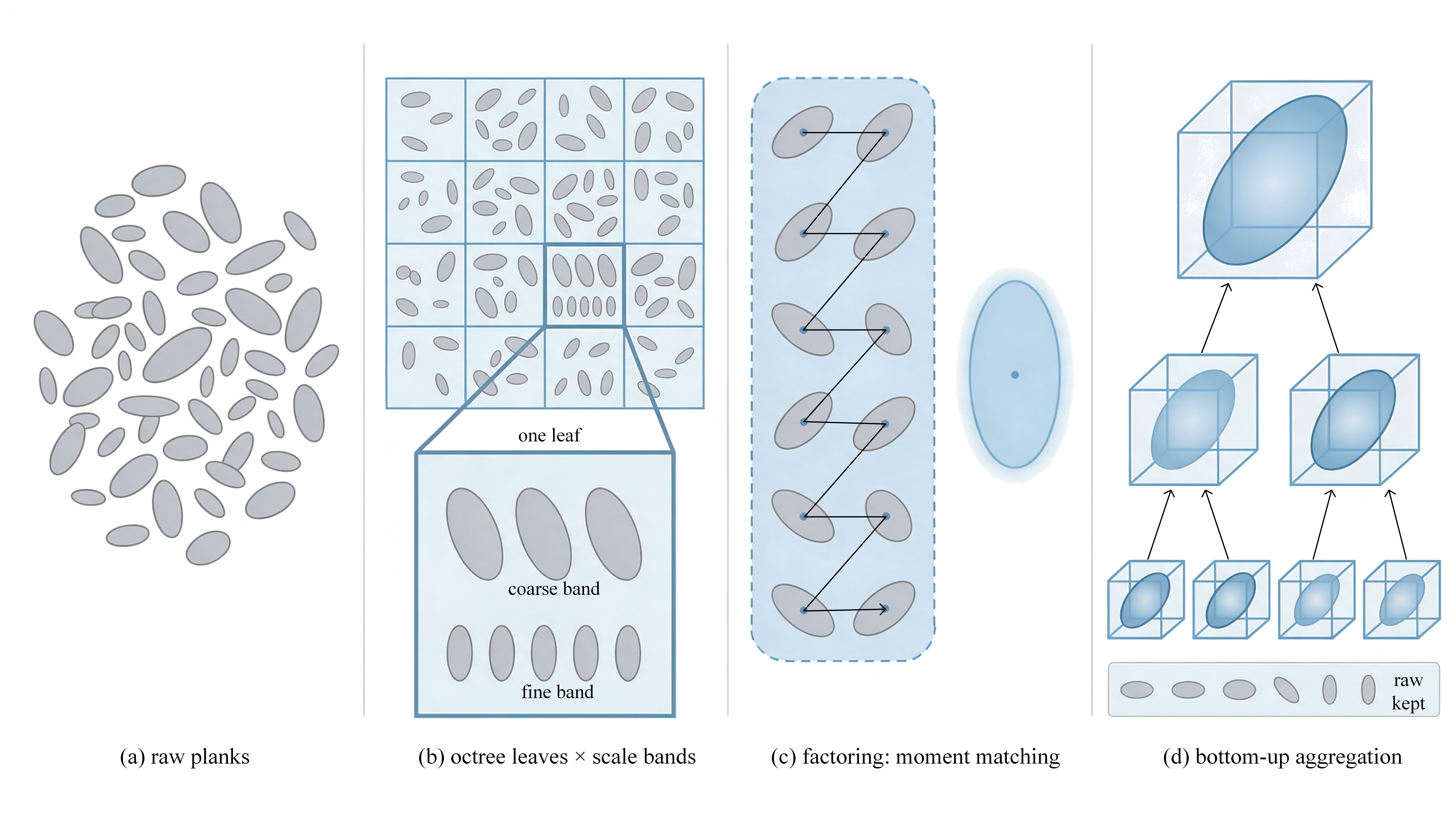}
\caption{Factoring-tree construction. (a) The trained model is a flat
set of raw planks across scale bands. (b) A GPU octree partitions space
into leaves (pure spatial indexing); within each leaf, planks split into
scale bands $\ell=\lfloor\log_2(s/s_{\mathrm{base}})\rfloor$. (c) Within
each (leaf, band), planks are Morton-ordered and cut into consecutive
buckets; each bucket is \emph{factored}: one equivalent
plank conserves weighted volume $W$, centroid $c^*$, and inertia $T^*$
(SH colors enter as weighted first moments). (d) Internal nodes
aggregate bottom-up, so every (node, band) owns one equivalent plank;
raw planks stay available in Morton order for the finest tier.}
\label{fig:factoring}
\end{figure*}

\begin{theorem}[Weighted covering---the collapse license]
\label{thm:collapse}
\textup{(i) (Unweighted; \cite{wz2025}, Thm.~1.2.)} For every
$\varepsilon>0$ there is $K>1$ such that, for all $\delta>0$ small: if
$\T$ satisfies \eqref{eq:ktcw} with $C\upre1$ and every
$a\times b\times2$ prism contains at most $100ab\delta^{-2}$ tubes of
$\T$, and $Y(T)\subset T$ has $\abs{Y(T)}\ge\lambda_0\abs{T}$, then
\[
\abs{\bigcup_{T\in\T}Y(T)}\;\ge\;\delta^{\varepsilon}\lambda_0^{K}
\sum_{T\in\T}\abs{T}.
\]
\textup{(ii) (Weighted, rendering form.)} If each tube carries a weight
function $w_T\colon T\to[0,1]$ of per-tube mass
$\abs{T}^{-1}\!\int_Tw_T\ge\lambda(w)\succap1$, then the saturated
weighted union volume $|U|_w:=\int\min\bigl(1,\sum_Tw_T(x)\bigr)\dd x$
obeys
$|U|_w\ \gtrsim\ \lambda(w)^{K+1}\delta^{\varepsilon}\sum_T\int_Tw_T$:
thresholding costs only constants depending on $\varepsilon$, so the weighted form loses
nothing beyond the $\varepsilon$-loss of (i).
\end{theorem}

\begin{lemma}[Scale truncation]\label{lem:trunc}
For cartoon-class fields, atoms whose projected footprint is below the
pixel scale can be dropped or collapsed with bounded loss.
\end{lemma}

\begin{theorem}[Adaptive covering rate]\label{thm:rate}
For $f\in\mathcal{C}_K$ (piecewise $C^2/H^2$ cartoon fields) there is an
$N$-atom combination of plates and tubes with
$\norm{f-f_N}_{L^2}\le CK\,N^{-1/2}(\log N)^c$; conversely, some
$f^*\in\mathcal{C}_K$ forces every $N$-atom combination to obey
$\norm{f^*-f_N}_{L^2}\ge c_0K(N\log_2N)^{-1/2}$.
\end{theorem}

\begin{corollary}[Budget allocation]\label{cor:alloc}
Minimizing the geometric $+$ appearance error
$C_1KN^{-1/2}+C_2K'L^{-p}$ under the parameter budget
$M\asymp N(L+1)^2$ yields $N^*\asymp M^{p/(p+1)}$ and
$L^*\asymp M^{1/(2(p+1))}$; across scale bands $s_j=2^js_0$ the
band budgets decay exponentially, $m_j\propto2^{-j\gamma}$.
\end{corollary}

Theorem~\ref{thm:collapse} is the union-volume estimate of
Wang--Zahl~\cite{wz2025}, the theorem behind the 2025 resolution of the
three-dimensional Kakeya set conjecture, reinterpreted with continuous
opacity weights---to our knowledge the first rendering use of this
theory. It certifies when a family of rendering
primitives can be \emph{collapsed} into fewer representatives without
forfeiting more than an $\varepsilon$-fraction of volumetric mass, and
its one-line corollary caps the packing density of any convex region at
$C\Delta_{\max}(\T)$ (the Katz--Tao constant of Definition~\ref{def:wolff}),
which is what keeps
the scale-band decomposition of Section~\ref{sec:lod} bounded. Together
with Zador's law \eqref{eq:zador}, the four statements answer, in order:
how much error is unavoidable (the floor), which groups are redundant
(collapse), what may be dropped outright (truncation), and how a fixed
budget should be split (allocation). Stated as questions (how many primitives does this scene need, how
much of it is redundancy, how should a fixed budget be split across a
city), the theory answers all three; the rest of the paper measures the
answers and builds methods on them.


\section{Budgeted-GS: training and deployment}\label{sec:lod}

The floor and the selection rules are operationalized two ways: at
training time, budget-centered training grows the model the floor
prescribes (Section~\ref{sec:bct}); at deployment time, the factoring
tree and its pixel-footprint scheduler serve, per view, exactly the
working set that fits the memory at hand (Sections~\ref{sec:lod-tree}
and~\ref{sec:lod-sched}).

Both halves of our system follow from one observation: a trained model is
a redundant, multi-scale arrangement of oriented boxes (each primitive
occupies a box of half-extents $h=3\sigma$ weighted by its opacity), and
budget control means removing that redundancy without visible loss. Three
principles govern every selection rule in this paper; they are the
operational form of the theorems stated in Section~\ref{sec:theory}.

The first principle is the \emph{collapse license}: a group of primitives that \emph{locally
over-covers} space, many overlapping boxes of comparable scale in one
region, can be merged into a few representatives with bounded
volumetric loss (Theorem~\ref{thm:collapse}; the moment-matching counterpart of
the optimal-atom prediction, Section~\ref{sec:theory}). This licenses
the moment-matched aggregates of the factoring tree
and opacity-ranked pruning inside redundancy-flagged regions.

The second is \emph{scale truncation}: a primitive whose projected footprint is
below the pixel scale can be dropped or collapsed with bounded image
error (Lemma~\ref{lem:trunc}; the U2 projection-error regime). This licenses the
pixel-footprint threshold $t=c\cdot d/f$ of the scheduler and the stopping
rule that freezes density control once pruning stops paying for itself.

The third is \emph{budget allocation}: under a fixed primitive budget, coarser
scale bands deserve exponentially more budget than finer ones (Corollary~\ref{cor:alloc};
the $K_js_j/d_j$ law); the same footprint rule instantiates this allocation
\emph{per view}. The three principles are grounded in recent union-volume estimates for
unions of convex sets~\cite{wz2025} and classical adaptive-approximation
rates~\cite{devore1998}; measured error rates on synthetic scenes match
the prediction within measurement noise.

\subsection{Budget-centered training}\label{sec:bct}

The floor answers ``how many primitives do I need at quality $Q$?'', and
the answer is also a training prescription. Traditional pipelines train a
large model and compress it: the primitives are optimized for a different
capacity, and the training time spent on the removed half is wasted. BCT
instead locates the knee $N^*$ from the measured curve (Fig.~\ref{fig:teaser}; the budget with
the largest steepest-window slope, Section~\ref{sec:measure}), then trains
the MCMC densifier \emph{directly at cap${}=N^*$}: the initial primitive count
is set to a fixed initial fraction of $N^*$ and the refine-stop fraction
to a second, both scaled by the measured
redundancy $\theta_{\mathrm{waste}}=N_{\mathrm{opt}}/N_{\mathrm{full}}$.

Standard 3DGS densification/pruning is driven by positional gradients and
opacity. We keep densification but replace the pruning signal by the
\emph{$\theta$ redundancy map} (the collapse license; supplementary
material):
$\theta$ is computed per region from the local factoring of the
model: regions with high local multiplicity (low $\theta$) are
\emph{over-covered} and are exactly where primitives can be removed with
the smallest approximation loss; pruning inside $\theta$-flagged regions
removes primitives by opacity. A \emph{rate-based stopping criterion}
freezes density control once validation improvement falls below
$\varepsilon=0.5$~dB.

\subsection{The factoring tree}\label{sec:lod-tree}

A trained model is frozen, but the hardware it must run on is not: the
same city has to serve a render farm, a desktop, a laptop, and a VR
headset, and at the far end it has to render in real time on a consumer
GPU. The floor's allocation law (budget should follow $K_js_j/d_j$ per
tile and per view) turns this into a selection problem. Concretely, applied per tile $j$ (regularity
$K_j$, extent $s_j$, distance $d_j$), Corollary~\ref{cor:alloc} yields
the \emph{level-selection law}
\begin{equation}
\small N_j^*(v)\propto K_j\sqrt{A_j(v)}\asymp\frac{K_j\,s_j}{d_j},
\qquad
e_j^2\asymp\frac{c^2\,K_js_j}{d_j}\cdot\frac{\sum_kK_k s_k/d_k}{N}\cdot
\occ_j .\normalsize
\label{eq:lodlaw}
\end{equation}
with $A_j(v)$ the projected screen area of tile $j$ in view $v$. For
uniform regularity it reduces to the classical screen-space-error
heuristic and adds what it lacks: rougher tiles draw more budget,
occluded tiles pay the visibility factor, and the per-tile error is
certificate-verifiable.

Given a trained model with $N$ primitives (means $\mu_i$, rotations
$R_i$, scales $\sigma_i$, opacities $\alpha_i$, SH colors), we first
convert each primitive to a \emph{plank} $(\mu, R, h{=}3\sigma, \alpha)$.
The scene bounding box is partitioned by a GPU octree (leaf capacity
4096); the
structure lives in what its nodes store (below), and within each leaf
the planks are split into \emph{scale bands}
$\ell=\lfloor\log_2(s/s_{\mathrm{base}})\rfloor$ where $s=\max h$ and
$s_{\mathrm{base}}$ is the 5\% quantile of $s$. Each (leaf, band) group is
then \emph{factored} into at most $256$ equivalent planks
by batch moment matching: within a group, planks are ordered by Morton
code (spatial locality), split into consecutive buckets, and each bucket is
replaced by a single plank conserving
\begin{equation}
\begin{aligned}
W&=\sum_i\alpha_i|V_i|,\qquad
c^*=\frac{\sum_i\alpha_i|V_i|\mu_i}{W},\\
T^*&=\sum_i\alpha_i|V_i|\big((\mu_i-c^*)(\mu_i-c^*)^{\top}+\tfrac13 R_i\,\mathrm{diag}(h_i)^2\,R_i^{\top}\big),
\end{aligned}
\end{equation}
i.e.\ weighted volume, centroid, and inertia tensor; the equivalent plank
is $(c^*,R^*,h^*,w^*)$ from the eigendecomposition of $T^*$ with
$w^*=W/|V^*|\in[0,1]$ (Fig.~\ref{fig:factoring}c). Colors enter as
opacity-weighted first moments of the
SH coefficients (DC and higher bands separately).

The internal nodes of the octree receive the same treatment bottom-up, so
every (node, band) has one \emph{aggregate plank}. The raw primitives are
retained in Morton order for the finest tier. Construction on
multi-million-primitive scenes takes seconds on one GPU, and the
aggregate representation is $119\times$ smaller than the input
(supplementary material).

\subsection{The pixel-footprint scheduler}\label{sec:lod-sched}
At render time, for camera position $p$ and focal length $f$ (pixels), each
(leaf, band) with mean scale $\bar s$ and centroid distance $d$ is assigned
to one of three tiers by the threshold $t=c\cdot d/f$:
\begin{itemize}
  \item \textbf{raw}: members individually resolvable ($\bar s\ge t$)
  $\Rightarrow$ render the original Gaussians;
  \item \textbf{fine}: the cluster resolvable but members not
  ($\bar s<t\le \bar h$, $\bar h$ = equivalent-plank scale) $\Rightarrow$
  render the fine equivalent planks;
  \item \textbf{coarse}: sub-pixel ($\bar h<t$) $\Rightarrow$ render the
  single aggregate plank of the (leaf, band).
\end{itemize}
$c$ is a single scalar \emph{budget parameter}: $c{=}1$ means \emph{exactly
resolvable} (the reference tier), larger $c$ progressively collapses
sub-$c$-pixel content (scale truncation; supplementary material), and $c<1$ keeps members below the pixel scale
(super-resolving regime); on distant aerial cameras most bands are
already sub-pixel at $c{=}1$, and our real-time renderer operates at
$c\approx0.3$ there (Section~\ref{sec:exp-rt}).


\section{Experiments}
\label{sec:experiments}

We validate the theory and both methods end to end: the floor
itself across 13 public scenes, then Budget-Centered Training, then
Factoring LOD from object scenes to city scale, closing with
ablations and efficiency.

\subsection{The floor, measured across 13 scenes}
\label{sec:measure}

The floor is measured with a preregistered
protocol (criteria frozen before data collection; all verdicts frozen
before reading results; 3 fixed seeds; JSON certificates at every stage,
supplementary material): on each scene, MCMC training at cap
$2\times10^6$ followed by opacity-prune to $N$ and fine-tune yields the
empirical floor curve $R(N)$ at 9 budgets; the preregistered
floor-intactness test asks that no scan point crosses the fitted floor
by more than $\log_{10}(0.8)$; published operating points are
re-anchored to their own baseline at matched capacity ratio with a
$5\%$ ($0.22$~dB) tolerance; and a duality certificate lower-bounds the
achieved error at every (budget, view).

On a synthetic cartoon control with the same
protocol, the chord lands directly on Zador's law ($-0.549/-0.548/-0.551$,
$R^2{=}0.997$). On the 13 public scenes (seed 0; 3 seeds on
garden/room/bonsai) chord slopes are seed-stable ($\pm0.004$--$0.011$)
but all fall outside the preregistered band $[-0.55,-0.45]$, a
registered negative. The reading-window audit explains it: the SfM point
clouds have TWO-NN intrinsic dimension $1.9$--$2.05$, and the
\emph{steepest} three-band windows reach $-0.41$ to $-0.53$ across
scenes and seeds (bonsai $-0.46/-0.49/-0.53$; room $-0.469\to-0.496$
when the fine-tune budget is quadrupled), straddling the predicted
$-1/2$. Floor intactness holds $13/13$; the duality certificate holds at
every (band, view) on room with tightness $0.28\to0.03$
(Table~S10, full data in
supplementary material).

The published frontier is flat: every compression method with a
same-paper $(N,\mathrm{PSNR})$ pair defines a resolution-invariant
two-point trajectory (Fig.~S3). Across Compact3DGS and EAGLES on all 13 scenes
plus Mini-Splatting~\cite{fang2024minisplatting} and
RadSplat~\cite{niemeyer2024radsplat} at dataset level, these slopes occupy
$[-0.28,+0.09]$ with mean $-0.04$: removing $60$--$70\%$ of a converged
model's primitives costs under $1$~dB across the surveyed compressors,
while the floor law (asymptote $-1/2$, reached end-to-end on the synthetic
control and bracketed by our steepest windows) is an order of magnitude
steeper. Published compression operates deep in the over-budgeted regime:
its reported gains ride the redundancy
$\theta_{\mathrm{waste}}=N_{\mathrm{opt}}/N_{\mathrm{full}}$ that the
floor makes measurable, not proximity to capacity. The universal
(model-free) colour-quantization floor, a lower bound on
\emph{every} representation at budget $N$ computed from test images
alone, confirms it: all 26 published operating points lie above the band
(Fig.~S4).

Three questions therefore have checkable answers: \emph{certified
compression headroom} (measure the floor once, place any candidate point
on it; with our measured estimator looseness folded in, headroom is
certified on eight published points, up to $+1.9$~dB);
\emph{when lossless is impossible} (the gap decomposition separates
structural coupling loss from source capacity); and \emph{the right
budget}: the floor knee $N^*$ is the budget where the steepest-window
slope is largest, which is exactly where BCT trains
(Section~\ref{sec:bct}).

\subsection{Training-side results}\label{sec:bct-results}

The redundancy signal and the stopping rule behave as the floor
predicts across the synthetic and Mip-NeRF 360 scenes, under the same
preregistered discipline as the floor study (3 fixed seeds on
garden/room/bonsai, seed~0 elsewhere; matched near-$N$ runs).

On the 13-scene study,
BCT is faster and better in the same-ruler sense at the budgets where it
is asked to exceed a train-then-prune point: at the matched
compact3DGS room budget, $529$k, BCT reaches $31.94\pm0.04$~dB
(three seeds, re-measured protocol) against the published
Compact3DGS ($30.23$~dB; the original
high draw was $31.20$). An earlier protocol generation of the same
comparison makes the capacity point directly: two audit runs agree on
$31.21$~dB at cap $670$k, where that generation's prune pipeline needs
$1$M primitives for the same $31.21$~dB---BCT spends $33\%$ fewer
primitives for the same rendering quality. (All Table~\ref{tab:bct}
room rows are re-measured under the current protocol; the audit
generation is kept protocol-separated.) The same-ruler comparison over six scenes (Table~\ref{tab:bct})
is the honest version of this claim: at counter BCT beats both published
points ($+0.24$/$+0.58$~dB) and our prune scan ($+1.90$~dB vs.\ the
$500$k scan point), at bonsai
$640$k it beats EAGLES ($+0.86$), and at kitchen $1131$k it beats
Compact3DGS by $+1.70$~dB; only at flowers does it stay below the
published compressors ($-0.29$/$-0.71$~dB), while still leading our own
prune pipeline there ($+0.44$/$+0.65$~dB); and at garden both budgets
beat the published points ($+0.24$/$+0.44$~dB). The honest summary: BCT
leads our train-then-prune pipeline at every matched budget, and it now
beats the published compressor points on five of six scenes (all but
flowers). The single loss is diagnostic rather than accidental: flowers
is the only scene whose measured knee sits at the full scanned budget
($\theta_{\mathrm{waste}}{=}1.0$, Section~\ref{sec:measure})---source-limited,
not redundancy-limited---so there is nothing for budget selection to
spend there. Its decisive advantage
is budget selection at the floor knee (33\% fewer primitives at matched
quality) plus a single 78-minute run instead of the 175-minute
train--prune--fine-tune route (measured on room: 125~min training at
$1$M + 50~min prune$+$fine-tune). An extension to the outdoor scenes of the study
(bicycle, stump, treehill), where the published points sit highest above
the floor, appears in Table~S4.

Table~\ref{tab:bct} records the tuned variants: in all 7 runs (across
three scenes and two budgets each, the initial-fraction/stop
schedule variants) the original recipe (SfM init of $100$k primitives,
refine to convergence) is best or tied, a negative ablation confirming
that BCT needs no hyper-parameter search, and that the gains above are
not tuning artifacts. Timing (measured on room, shared A100 under neighbor load; ratios are
same-card): BCT $78\pm7$~min over three seeds; the traditional route
for the same budget (train at $1$M, prune, fine-tune) $\approx$175~min,
$2.2\times$ slower for lower quality (Table~\ref{tab:bct}).

\begin{table*}[t]
\centering\small
\begin{tabular}{llllrrr}
\toprule
scene & budget & schedule & PSNR~dB & ours-prune & published \\
\midrule
room & 529k & orig & \textbf{31.94}$\pm$.04 & 30.14 & 30.23 (C3DGS) \\
room & 670k & orig & 31.95 & 31.21 @1M & 31.47 (EAGLES) \\
room & 1M & orig & \textbf{32.24} & 31.21 & --- \\
\midrule
garden & 1650k & orig & \textbf{27.15}$\pm$.04 & 25.98 @1M & 26.91 (EAGLES) \\
garden & 2209k & orig & \textbf{27.25}$\pm$.06 & --- & 26.81 (C3DGS) \\
\midrule
bonsai & 601k & orig & \textbf{32.13} & 29.38 & 32.08 (C3DGS) \\
bonsai & 640k & orig & \textbf{32.18} & --- & 31.32 (EAGLES) \\
\midrule
counter & 536k & orig & \textbf{28.95}$\pm$.02 & 27.05 @500k & 28.71 (C3DGS) \\
counter & 560k & orig & \textbf{28.98}$\pm$.02 & --- & 28.40 (EAGLES) \\
\midrule
flowers & 1330k & orig & 20.58$\pm$.02 & 19.93 @1M & 21.29 (EAGLES) \\
flowers & 1525k & orig & 20.60$\pm$.01 & 20.16 @2M & 20.89 (C3DGS) \\
\midrule
kitchen & 1M & orig & \textbf{31.97}$\pm$.10 & 30.60 & 30.77 (EAGLES) \\
kitchen & 1131k & orig & \textbf{32.18}$\pm$.07 & --- & 30.48 (C3DGS) \\
\bottomrule
\end{tabular}
\par\vspace{3pt}
\caption{BCT vs.\ the same-ruler train-then-prune pipeline and published compressor baselines.}
\label{tab:bct}
\end{table*}

\subsubsection{Outdoor extension}

The same-ruler BCT comparison extended to the three outdoor scenes of
the study (bicycle, stump, treehill) at the published compressor
budgets, plus the frontier reading of the outcome. Protocol identical
to Table~\ref{tab:bct} (seed 0, same eval, original recipe: SfM init
$100$k primitives, refine to convergence). ``ours-prune'' = the matched-budget
point of the same-ruler train-then-prune scan of
Section~\ref{sec:measure} (nearest scan band; interpolated in $\log N$
to the tabular budget).

The outdoor extension (Table~S4) makes the floor's
practical value testable in
both directions. BCT, which spends the floor's answer \emph{directly
at a budget}, lands on the same curve as our train$\rightarrow$prune
$\rightarrow$fine-tune scan (differences $\le 0.15$~dB); the two
procedures agree because both consume the budget the floor says the
scene can carry. The published compressor points at the same budgets
are $0.8$--$2.9$~dB \emph{above} that curve. Their gains therefore do not come from capacity: on stump, where the published points are
$2.7$--$2.9$~dB above our measured floor at matched budgets, the
measured training-time redundancy $\theta_{\mathrm{waste}}$ fully
accounts for the reported gap (frontier analysis,
Section~\ref{sec:measure}).


\subsubsection{Budget-aware density control at the ceiling}\label{sec:bct-lego}

Table~S5, Fig.~S2, and
Fig.~S14a: on lego our
budget-aware training with the BCT stop rule (Section~\ref{sec:bct})
reaches 36.85~dB with 151k primitives: the same quality as Vanilla with
339k ($44\%$ of the budget, $-0.34$~dB), and $+0.23$~dB over MCMC at a
$12\%$ smaller budget. The full Pareto curve (Fig.~S2)
dominates MCMC at every tested budget. Perceived, not just measured:
Fig.~S5 shows the two models side by side on a
held-out view with the residual amplified $10\times$; on this view the
151k model is \emph{better} than Vanilla (24.88 vs.\ 21.46~dB; the
aggregate mean is $-0.34$~dB over the full test set), and the largest
residuals are structural edges rather than texture (mean
$|\Delta|=0.013$).

\subsubsection{Generalization}\label{sec:bct-gen}

Table~S6 generalizes the stop rule to four further
scenes (chair, ship, bonsai, room).

Post-hoc, the redundancy signal keeps its value on the indoor scenes
(seed-0 runs): pruning the trained model by the BCT stop rule
dominates random pruning at every matched budget on bonsai and
room; e.g.\ at 70\% kept, a 32-view 3DGS re-render of the same
checkpoints measures $31.10$/$35.64$~dB for $\theta$-pruning
vs.\ $28.10$/$29.56$~dB random ($+3.0$/$+6.1$~dB).
Fig.~S6 shows the corresponding renders: at matched
70\% budget the $\theta$-guided models stay close to their full renders,
while random pruning visibly degrades foliage and edges.

The contribution-scored export of
the score export of Section~\ref{sec:lod} is a pure ranking of an already-trained
model, so it applies to any primitive set. Table~S7
applies it to the trained lego/bonsai/room models and the two Mip-NeRF 360
indoor scenes, evaluated by the same CityGaussian test protocol that
produced the training-side numbers. On every scene and every budget, the
score ranking beats opacity-ordered pruning, the baseline the same
protocol ships, by $+2.7$ to $+7.9$~dB (PSNR), with large margins also
in SSIM and LPIPS. At $50\%$ the margin costs almost nothing: $-0.13$ and
$-0.54$~dB on room and bonsai versus the full model. At $25\%$ the loss
is substantial ($3.1$--$13.9$~dB), which is the expected behavior on
\emph{already-compressed} models: these trained models are down to
36--66\% of Vanilla, so they carry less removable redundancy than
Block\_A (where $24\%$ of primitives are invisible from every training
view and the curve is flat to $25\%$). The clean division of labor: the
score export recovers what a post-hoc pass can, and the training-side
BCT is the tool where a fixed quality target must be met.


The two stages also compose. On room, we build factoring trees over
two roots, the full model (1.66M primitives) and the budget-centered
model trained at the floor's knee (529k), and deploy both through the
same scheduler at a matched budget of ${\approx}520$k rendered
primitives (34 held-out views, LOD render protocol). The BCT-born tree
renders 34.19~dB, within 0.16~dB of its own root; the full-born tree,
collapsed to the same budget, renders 22.83~dB. Budget-centered
training is thus the correct front end for factoring-tree deployment:
a model born at the floor's budget passes through the tree essentially
unchanged, because it contains almost no redundancy for the tree to
lose (per-tier data in the supplementary material).

\subsection{The LOD quality window}\label{sec:exp-lod}
\begin{table*}[t]
\centering
\small
\begin{tabular}{lrrr}
\toprule
config & PSNR & \#primitives & $\Delta$PSNR \\
\midrule
full & 29.97 & 5.78M & --- \\
LOD c$=1$ & 29.31 & 94.4\% & $-0.66$ \\
LOD c$=2$ & 26.89 & 83.0\% & $-3.08$ \\
LOD c$=4$ & 22.56 & 62.4\% & $-7.41$ \\
adaptive proxies & \textbf{29.80} & --- & \textbf{$-0.17$} \\
\bottomrule
\end{tabular}
\caption{garden (Mip-NeRF 360), trajectory views: fixed threshold vs.\ residual-proxy adaptive selection.}
\label{tab:garden}
\end{table*}

\begin{table*}[t]
\centering
\small
\begin{tabular}{lrrr}
\toprule
config & PSNR (med) & model data & $\Delta$PSNR \\
\midrule
\multicolumn{4}{l}{\emph{our LOD tiers, rendered by the 3DGS rasterizer:}} \\
full & 28.38 & 1240.6 MB & --- \\
LOD c$=4$ & 28.33 & 844.1 MB & $-0.05$ \\
LOD c$=8$ & 26.26 & 552.7 MB & $-2.12$ \\
LOD c$=16$ & 23.72 & 288.9 MB & $-4.66$ \\
LOD c$=32$ & 20.80 & 110.6 MB & $-7.58$ \\
LOD c$=64$ & 17.46 & 33.9 MB & $-10.92$ \\
\midrule
\multicolumn{4}{l}{\emph{the same LOD tiers rendered through the GPS renderer}~\cite{rijsdijk2026gps}:} \\
full & 28.50 & 1926.2 MB & --- \\
LOD c$=4$ & 28.52 & 1310.5 MB & $+0.02$ \\
LOD c$=8$ & 26.88 & 858.1 MB & $-1.62$ \\
LOD c$=16$ & 24.13 & 448.6 MB & $-4.37$ \\
LOD c$=32$ & 20.88 & 171.7 MB & $-7.62$ \\
LOD c$=64$ & 17.38 & 52.6 MB & $-11.12$ \\
\midrule
\multicolumn{4}{l}{\emph{random-subsampling control through the same path:}} \\
random 44.6\% & 21.42 & 859.3 MB & $-7.08$ \\
random 23.2\% & 16.67 & 447.3 MB & $-11.83$ \\
random 8.9\% & 12.90 & 171.7 MB & $-15.60$ \\
\midrule
\multicolumn{4}{l}{\emph{Extended metrics, all 154 test views, mean (PSNR/SSIM/LPIPS), 3DGS rasterizer:}} \\
full & 27.71 & 1240.6 MB & --- \\
LOD c$=8$ & 26.10 / .747 / .355 & 552.7 MB & $-1.61$ \\
LOD c$=16$ & 23.45 / .693 / .413 & 288.9 MB & $-4.26$ \\
LOD c$=32$ & 20.52 / .627 / .474 & 110.6 MB & $-7.19$ \\
LOD c$=64$ & 17.51 / .548 / .539 & 33.9 MB & $-10.20$ \\
\bottomrule
\end{tabular}
\caption{MatrixCity Aerial: the paired blocks use 24 views and median
PSNR; the extended-metrics block uses all 154 test views and the mean.
Same model, cameras, and GT throughout.}
\label{tab:matrixcity}
\end{table*}

On garden (Table~\ref{tab:garden}), the residual-proxy adaptive
selector recovers $0.49$~dB over the fixed threshold at a comparable
primitive count ($1.05\times$).

On the city-scale model the near-lossless window is even wider
(Table~\ref{tab:matrixcity}): $c{=}4$ renders 68\% of the primitives at
$-0.05$~dB (visual comparisons: Fig.~S12 for the garden
tiers, Fig.~S10 for the same ladder on room). This
asymmetry (garden vs.\ city aerial), quantified by the orbit sweep of
Fig.~S7, matches the theory:
the aerial model is already tessellated at its training viewing distance,
so sub-pixel redundancy is exhausted at small $c$; the surplus is in
\emph{super-pixel} collapse. Across the full ladder the single model
sweeps a continuous memory--quality curve (Fig.~S15);
re-rendering the same tiers through the GPS renderer reproduces it
within $0.5$~dB ($c{=}4$ is lossless under both renderers), and the
${\sim}20.8$~dB operating point costs $110.6$~MB against GPS's full
$1926$~MB ($17\times$; same-budget numbers and byte-level accounting
in the supplementary material).

\subsection{Comparison with training-coupled LOD}\label{sec:exp-e1}

The most closely related systems are LOD hierarchies \emph{coupled to
training}: Octree-GS~\cite{ren2024octreegs} trains a multi-level anchor
structure and selects levels per view;
CityGaussian~\cite{liu2024citygaussian} and its
V2~\cite{liu2025citygaussianv2} partition large scenes and load blocks
by level at runtime. Our layer is complementary by construction: it is
built \emph{post-hoc} from any trained model, exposes a continuous
budget parameter, and is renderer-agnostic. Recent concurrent work reaches
consumer-GPU city scale by training the LOD itself, out-of-core streaming hierarchies~\cite{windisch2026alod} and LoD-aware training~\cite{kulhanek2025lodge}, so the axes are complementary rather
than competing: their budgets are fixed at training time, ours are chosen
per deployment on the same trained model. Table~S8 frames the
head-to-head on MatrixCity and makes the structural difference
concrete. Re-rendering the official Octree-GS Block\_all checkpoint
(8.01M stored anchors, each decoding to 10 neural Gaussians at render
time) under a held-out split reproduces its level curve
(27.23~dB at 145k visible anchors, degrading gracefully to
19.5~dB at 29k), a \emph{trained} representation whose level
set is fixed at training time. On the other side,
CityGaussianV2's runtime LOD cannot be evaluated post-hoc on the official
Aerial release: its renderer requires a per-partition partitioning file and one
trained checkpoint per partition block, artifacts of its partitioned
training that the published single checkpoint does not include. Our
post-hoc layer, in contrast, consumes that same 8.56M checkpoint
as-is and sweeps the continuous $c$ curve of
Table~\ref{tab:matrixcity} on it. Runtime memory tells the same story:
while rendering the same Block\_all views at $1600\times900$ through
each side's official path on the same 4080 SUPER, the VRAM peak is
$30.4$~GiB for Octree-GS (steady-state $15.9$~GiB) and $4.4$~GiB for
ours ($6.9\times$ lower at peak; Table~S8); the
decoded-per-frame neural representation is intrinsically memory-heavy,
independent of the fps gap.

\subsection{Pruning baselines}\label{sec:exp-prune}

The $\theta$$+$opacity rule's
city-scale behavior is measured on Block\_A (Table~S9):
the redundancy signal localizes sacrifice at city scale just as it does on
lego (Section~\ref{sec:bct-lego}).

At matched budgets,
Table~S9 compares two post-hoc exports of the same trained
model against LightGaussian~\cite{fan2024light},
RAP~\cite{yang2026rap}, and a faithful single-pass reimplementation of
PUP~3D-GS's Fisher sensitivity score~\cite{hanson2025pup} (diagonal
Fisher over means and scales, applied through the same export path):
the $\theta$$+$opacity rule (coverage redundancy) and the
contribution-scored export of Section~\ref{sec:lod} (render-aware
ranking; the score itself is LightGaussian's, faithfully reimplemented
and applied through our post-hoc export path). The score export matches
Vanilla exactly at $50\%$ and $25\%$ of the primitives and
\emph{dominates all three baselines on all three metrics at $10\%$}
($+0.45$~dB PSNR over LightGaussian, $+0.80$~dB over RAP, $+0.37$~dB
over the PUP score); one $234$~s scoring pass, no retraining. The axis
extends where no baseline reports numbers: $13.49$~dB at $5\%$,
$12.97$~dB at $2.5\%$, $12.53$~dB at $1.3\%$ of the primitives.
Absolute Block\_A values are low by protocol design (DC-only evaluation
against real photographs); the table's purpose is the \emph{ordering} of
methods under one identical pipeline, not absolute quality.
$\theta$$+$opacity keeps a complementary role: it is the best LPIPS
export at every extreme budget ($0.547$ vs.\ $0.557$ for the score
export at $10\%$; $0.601$ for RAP), consistent with its design of
removing perceptually redundant mass. The score export, in turn, keeps
street markings and facades legible at $10\%$ where LightGaussian blurs
and RAP collapses (Fig.~S9).

The scheduler itself emits a standard
primitive stream (means, quaternions, scales, opacities, SH), consumable
by any splatting rasterizer; the export pipeline and render-aware
scoring are described in the supplementary material.

\subsection{Real-time rendering at city scale}\label{sec:exp-rt}

\begin{figure*}[t]
\centering
\includegraphics[width=\linewidth]{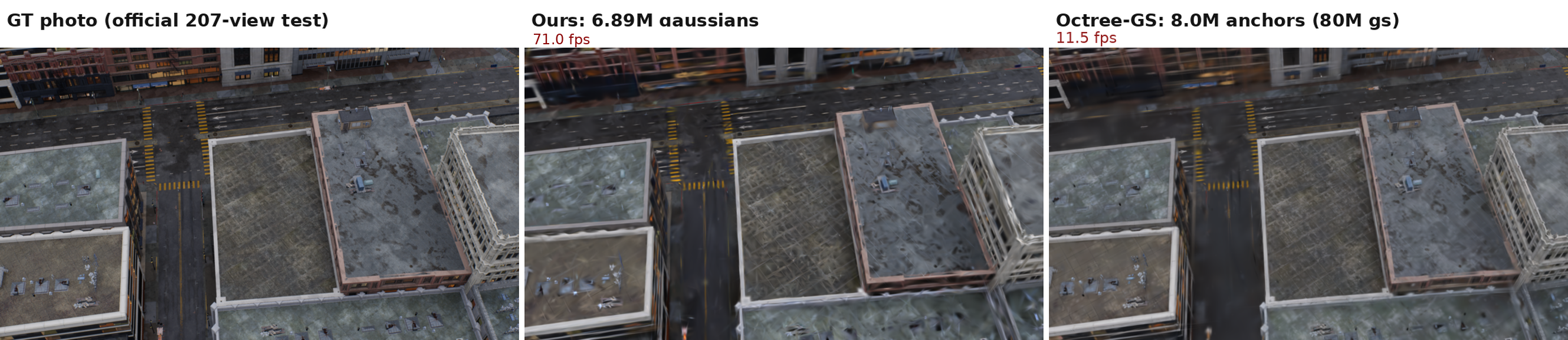}
\caption{City scale, official Block\_all 207-view test trajectory,
same view and eval protocol throughout: GT photo; our $6.89$M model at
$71$~fps; the
re-rendered official Octree-GS checkpoint ($8.0$M anchors, $80$M
decoded Gaussians) at $11.5$~fps.}
\label{fig:cityscale}
\end{figure*}
\begin{figure*}[t]
\centering
\includegraphics[width=\linewidth]{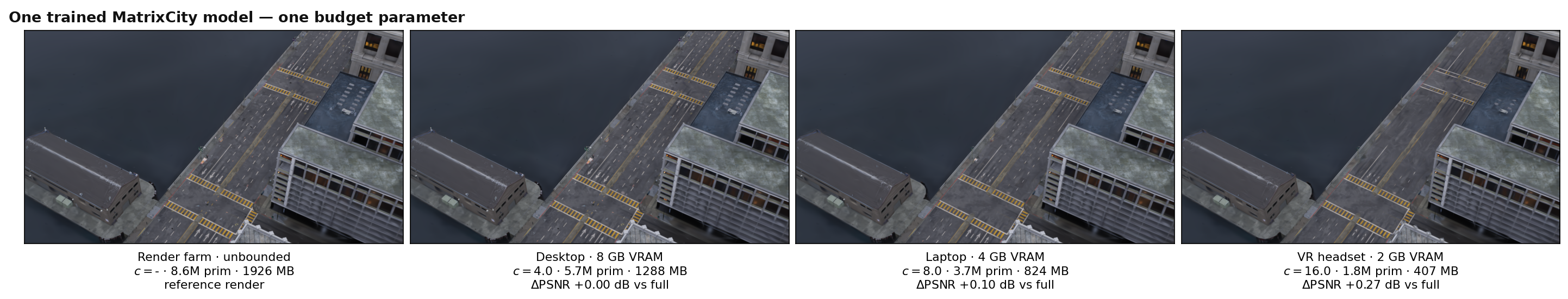}
\caption{One trained MatrixCity Aerial model, four deployment budgets chosen by the single parameter $c$ (all panels: one RTX 4080 SUPER, native $1080$p).}
\label{fig:deployment}
\end{figure*}

The two sides of the pipeline compose at city scale on our v3 retrain of the
official Block\_all data ($6.89$M primitives, full SH;
Fig.~\ref{fig:cityscale}): build the factoring tree
once, then render interactively under a capacity budget. The same
discipline serves deployment: the trained MatrixCity Aerial model
(8.56M primitives, full SH) renders at four memory budgets chosen by
the single parameter $c$ at native $1080$p on one RTX 4080 SUPER
($c{=}4/8/16$: 1288/824/407~MB on 8/4/2~GB-class targets); quality is
visually preserved down to the 2~GB tier ($|\Delta\mathrm{PSNR}|\le
0.27$~dB vs.\ full on this held-out view), with no retraining or
re-export (Fig.~\ref{fig:deployment}).

We verify
that the budget semantics carry to an interactive runtime, and we measure
where the real-time boundary actually lies on our v3 retrain of the
official Block\_all data
($6.89$M primitives, full SH, $25.43$~dB on the official 207-view protocol
with the 3DGS rasterizer; Table~\ref{tab:rt}). The deployment layer is CP-LOD
(coverage-preserving prefix LOD): per frame a leaf-only cut of the factree
is traversed on GPU, a \emph{capacity} budget $K$ bounds the materialized
working set, Morton-strided subsampling with scale$\times\sqrt{r}$ and
opacity $1{-}(1{-}\alpha)^r$ compensation preserves coverage at any budget
(Table~\ref{tab:rt}, Fig.~\ref{fig:rt_ladder}b). Selection and gather run
entirely on GPU---the selected prefix is collected straight into the
rasterizer buffers with no CPU-side collection or host-to-device
upload---and the GPU selection reproduces the offline reference exactly,
so the runtime inherits the evaluation semantics measured above; on the
official 207-view trajectory the per-frame selected count changes by a
median of ${\sim}1\%$ between consecutive views, and residual popping is
confined to tier boundaries (Limitations). Three properties matter
for deployment. First, selection is resolution-invariant by construction:
the materialized set is \emph{identical} at $800^2$ and $1920\times1080$
at every budget (measured equal counts on both); the naive
threshold-schedule drift of the tier prototype is gone. Second, the
full set renders in real time: at native $1920\times1080$ with full view-dependent SH
on one RTX 4080 SUPER, the official 207-view test trajectory renders at
$71$~fps mean (full set, full SH, full coverage; same protocol
as the $25.43$~dB measurement; timing drops the first 10 frames, and
the per-level fps curve is Table~S8); the orbit probe renders the full $5.99$M at
$49.5$~fps; at $800^2$ the same full set
runs $79$--$198$~fps. As a commodity-8~GB reference, the DC-only compact
path renders the original $19.93$M Block\_all model at $30.7$~fps at
$800^2$ on an RTX 3060 Ti (supplementary material). Third, the capacity ladder is smooth: at $1080p$,
$0.34$M/$0.68$M/$1.38$M/$2.59$M selected render
$76.8$/$65.0$/$59.0$/$57.0$~fps at $18.1$/$20.3$/$21.7$/$23.4$~dB
self-PSNR against the full render, so the budget--quality trade-off is
continuous and deploy-time (Fig.~S14b).

The re-rendered official Octree-GS
Block\_all checkpoint on the \emph{same} 207 views at
$-0.30$~dB resolution-matched model quality ($25.43$ vs.\ $25.73$~dB,
Octree-GS renders upsampled to $1080$p), both sides under each side's
own renderer: our rasterizer is $6.2\times$ faster ($71$ vs.\
$11.5$~fps). All fps numbers in this subsection are measured on the
same RTX 4080 SUPER. The fps contrast reflects architecture: the
anchor-based checkpoint stores $8.01$M anchors (decoding to
${\sim}80$M Gaussians, a $2.5$~GB model) but renders only a
LOD-culled ${\sim}0.9$M per view and pays a per-frame MLP decode,
whereas our $6.89$M static explicit splats ($1.7$~GB) feed the
rasterizer directly with no decode. In the literature,
city-scale interactive rendering has been reported on datacenter GPUs
with training-coupled, SH-reduced
LOD~\cite{liu2024citygaussian,ren2024octreegs}, and more recently on
consumer GPUs by training the LOD representation itself, out-of-core streaming hierarchies~\cite{windisch2026alod} and LoD-aware training~\cite{kulhanek2025lodge}, in each case for budgets baked in at
training time; our layer is the post-hoc counterpart: one already-trained
model, full SH, and a budget chosen at deploy time on a consumer GPU.

\begin{table*}[t]
\centering
\small
\begin{tabular}{lrrrr}
\toprule
operating point & selected & fps & p5 fps & GPU mem \\
\midrule
capacity $K{=}0.5$M & 342k & 76.8 & --- & $3.9$ GB \\
capacity $K{=}2$M & 1.43M & 59.0 & --- & $4.3$ GB \\
full set (cap 16M) & 5.99M & 49.5 & --- & $7.8$ GB \\
\midrule
full set, official 207-view traj. & 6.89M & \textbf{71.0} & 55.2 & $7.8$ GB \\
$K{=}2$M, official 207-view traj. & 1.43M & \textbf{66.2} & 42.7 & $4.3$ GB \\
full set, $800^2$ & 5.99M & 79.1--198 & --- & --- \\
\midrule
Octree-GS full ckpt (207 views) & $\sim$0.9M/view & 11.5 & --- & --- \\
\bottomrule
\end{tabular}
\caption{Interactive deployment on Block\_all (v3 retrain, native $1080$p).}
\label{tab:rt}
\end{table*}
\begin{figure*}[t]
\centering
\includegraphics[width=0.98\linewidth]{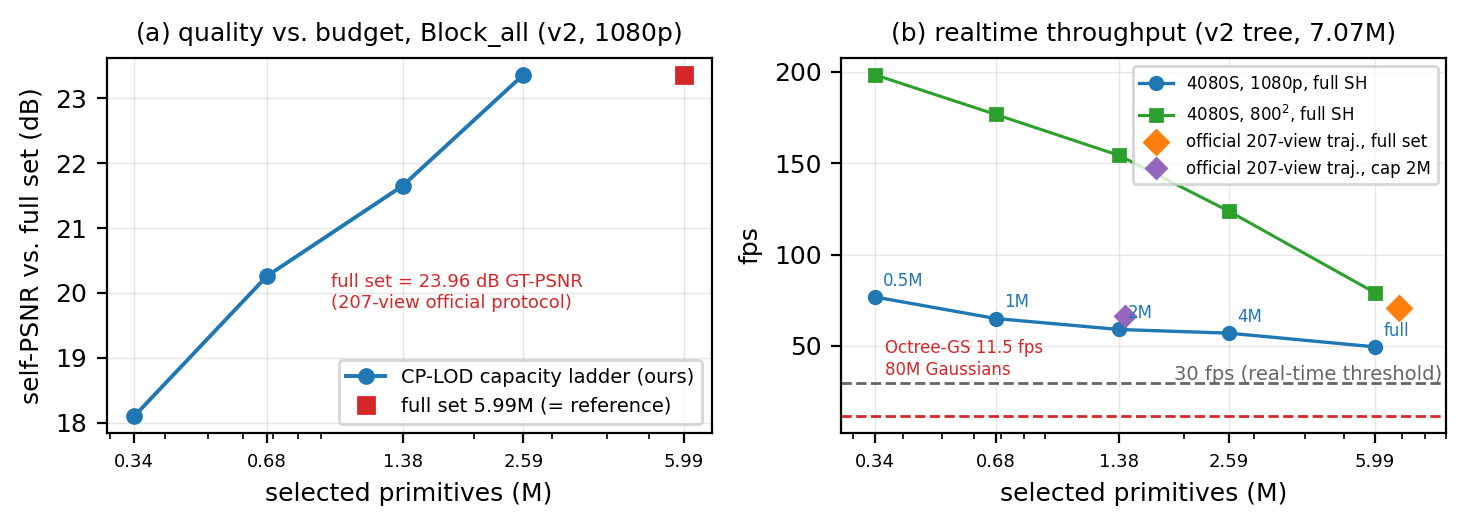}
\caption{Block\_all capacity ladder (a) and realtime throughput (b) on the v2 tree (7.07M, SH3).}
\label{fig:rt_ladder}
\end{figure*}

\subsection{Ablations}\label{sec:exp-abl}

On the training side (lego), replacing the redundancy-localizing
$\theta$ map by a random region prior costs 11--13\% more primitives at
the same PSNR; the $\theta$ signal's win is compression, not
accuracy. Without the rate-based stopping criterion, the default
densification controller over-prunes and loses 0.56~dB; with it,
training freezes automatically at the sweet spot (step 5000), and a
stop-only variant (no $\theta$ pruning) still reaches 37.00~dB at
216k (63\% of vanilla, $-0.19$~dB), isolating the stopping rule as the
accuracy-critical component. Hyper-parameter sensitivity is
$\le0.10$~dB. On the rendering side, a DC-only tree loses $5$~dB
systematically (storing opacity-weighted SH first moments for all
bands is required), and the residual-proxy selection recovers
$+0.49$~dB over the fixed threshold at a comparable count (garden,
Table~\ref{tab:garden}).

\subsection{Efficiency}\label{sec:exp-eff}

Construction also fits the desktop GPU for the Block\_all trees
reported here: on the $10$M $\theta$-pruned
Block\_all tree (full SH moments) it peaks at $6.45$ GiB on the 8 GB
RTX 3060 Ti ($2.53$ GB output); the $7.07$M Block\_all v2 tree
(full SH, $1.87$ GB) builds on the same consumer card. (The MatrixCity
Aerial tree is a one-time offline build.)


\section{Limitations}\label{sec:limit}
We evaluate on a single GPU; transparent and occluded-region extremes
are not covered, and memory streaming (GPS's native regime) remains
future work. The proved layer covers the compositing renderer family,
atomic and field representations, and the consistency functional with
matching upper and sharp lower bounds; three items remain open---the
sharpness of the visibility-exponent bound, SGD trainability of
implicit fields at the proved $N^{-1/2}$ rate, and non-compositing
renderers such as caustics and participating media. The measured
$-1/2$ bracket holds on all 13 scenes inside the registered reading
window (indoor scenes read $-0.457\approx-1/2$).
\section{Conclusion}\label{sec:concl}
We began with the obstacle that blocks adoption: city-scale 3D
Gaussian Splatting scenes do not fit a consumer GPU, and no principled
answer existed for what could be given up. This paper proposed such an
answer: a measurable capacity floor for 3DGS derived from optimal
transport---rendering and training become transport problems, the Zador
law fixes where the floor lies, and three selection rules (collapse,
truncation, allocation) turn the floor into executable permissions. On
this foundation we built two methods: budget-centered training grows a
model directly at the floor's prescribed budget instead of training
large and pruning, and Factoring LOD turns one trained model into a
deployment-time knob that renders a full city in real time on a
consumer GPU across continuous memory budgets. Future work includes
memory streaming for out-of-core scenes, hysteresis and perceptual
validation for tier transitions, and extending the phase-space
formulation to time and to further primitive families. One floor, two
methods: deployment renders exactly what fits; training grows exactly
what the scene needs.
\bibliographystyle{eg-alpha}
\bibliography{eg2027_merged_refs}

\clearpage
\setcounter{section}{0}
\setcounter{theorem}{0}
\setcounter{figure}{0}
\setcounter{table}{0}
\setcounter{equation}{0}
\renewcommand{\thetheorem}{\arabic{theorem}}
\renewcommand{\thefigure}{S\arabic{figure}}
\renewcommand{\thetable}{S\arabic{table}}
\renewcommand{\theequation}{S\arabic{equation}}
\renewcommand{\thesection}{S\arabic{section}}

The framework definitions and the theorem-chain statements are
made in Section~4 of the main paper; the proofs the main text does
not carry live here. The covering formalism behind the design
principles, with full proofs of the main-text theorems, is in the
Background-and-guarantees section; the measurement protocol and
records behind Section~6.1 of the main paper are in the Measurement
appendix; notation is collected in the notation table.
\section{Proofs for the capacity-floor chain}
\label{app:sec-theory}

The framework definitions and the theorem-chain statements are
in Section~4 of the main paper; this appendix collects the proofs of the
load-bearing statements: full proofs for the loss identity, the duality
certificate, the transport theorem, the Zador construction, and the
solver-descent and taxonomy statements; proof sketches remain only for
the visibility-cell rate and the new-view bound; the deep imports
(Zador's law, the Wang--Zahl covering estimate, DeVore's width bounds)
are cited to their classical sources.

\begin{proof}[Loss identity $\Lcal=\Ccal-\Vcal$]
For the normalized coupling $\tilde\pi_p$ the conditional expectation
$\int \rho\,\dd\tilde\pi_p$ is the rendered colour $I(p)$. Expanding the
square in the transport cost and summing over pixels with weights $S(p)$,
\[
\begin{aligned}
\Ccal&=\sum_pS(p)\!\int\!\norm{\rho-I(p)+I(p)-c^{\gt}(p)}^2\dd\tilde\pi_p\\
&=\underbrace{\sum_pS(p)\!\int\!\norm{\rho-I(p)}^2\dd\tilde\pi_p}_{\Vcal}
+\underbrace{\sum_pS(p)\,\norm{I(p)-c^{\gt}(p)}^2}_{\Lcal},
\end{aligned}
\]
the cross term vanishing because $\int \rho\,\dd\tilde\pi_p=I(p)$. Rearranging
gives $\Lcal=\Ccal-\Vcal$. \qed
\end{proof}

\begin{proof}[Certificate part of the representation theorem]
Kantorovich duality gives
$\Ccal_{\mathrm{free}}(\theta)=\sup_{\varphi\oplus\psi\le c}
\bigl[\int\varphi\,\dd\mut+\sum_pS(p)\psi(p)\bigr]$
(Fenchel--Rockafellar; the constraint set is linear in the source
variable), so every dual-feasible pair lower-bounds
$\Ccal_{\mathrm{free}}$. The rendering coupling is feasible for the free
problem with the same marginals, so $\Ccal(\theta)\ge\Ccal_{\mathrm{free}}(\theta)$.
Substituting the loss identity $\Lcal=\Ccal-\Vcal$ proved above and
rearranging yields the duality-certificate display of the main paper. \qed
\end{proof}

\begin{proof}[Zador law in U5: construction, and imported lower bound]
\textbf{Upper bound (construction).} Partition the $d_{\mathrm{ps}}$-dimensional
support into cells of diameter ${\sim}\varepsilon$; cell $j$ has measure
$m_j\approx\lambda_j\varepsilon^d$. At the optimum the per-point error
is equalized across occupied cells, which gives cell $j$
$k_j\propto m_j^{d_{\mathrm{ps}}/(d_{\mathrm{ps}}+2)}$ points; hence
$N\asymp\sum_jm_j^{d_{\mathrm{ps}}/(d_{\mathrm{ps}}+2)}\asymp\varepsilon^{-2d_{\mathrm{ps}}/(d_{\mathrm{ps}}+2)}
\norm{\lambda}_{d_{\mathrm{ps}}/(d_{\mathrm{ps}}+2)}$, while the total squared error is
$\Rcal\asymp\varepsilon^2$. Eliminating $\varepsilon$ gives
$\Rcal(N)\le C_{d_{\mathrm{ps}}}\norm{\lambda}_{d_{\mathrm{ps}}/(d_{\mathrm{ps}}+2)}N^{-2/d_{\mathrm{ps}}}$
for the rendering instance.
\textbf{Lower bound (imported, under view completeness).} The matching
lower bound with the sharp constant is Zador's theorem
itself~\cite{zador1982,graf2000foundations}; the entropy
route~\cite{kolmogorov1961entropy} gives the cruder
$\varepsilon\ge c\,N^{-1/d_{\mathrm{ps}}}$ on sets of positive density.
One bridge is used and flagged: Zador bounds the quantization error of
the source measure $\mus$, and the rendering instance inherits it only
if the view set separates the source (view completeness, the
single-source axiom read as an injectivity requirement on the target
map); the theorem statement in the main text carries this hypothesis.
The intrinsic phase-space support dimension $d_{\mathrm{ps}}$, not the ambient
dimension, governs the rate. On surface scenes $d_{\mathrm{ps}}{=}4$ (surface
position $\times$ incident-direction hemisphere), giving
$\Rcal\propto N^{-1/2}$. \qed
\end{proof}

\begin{theorem}[Convergence of solvers, and consistency]\label{thm:u67}
Gradient descent on $\Lcal=\Ccal-\Vcal$ with $L$-Lipschitz
gradient is monotone and every stationary point satisfies
$\nabla\Ccal(\theta^*)=\nabla\Vcal(\theta^*)$; sampling/Lloyd/meshing
attain their classical rates ($N^{-2/d}$ / stationary quantizer /
$N^{-1}$ on $C^2$ surfaces). Single-source
cross-view colour agreement is structural: masked visibility-cell
atoms achieve the optimal rate
$\max_j\norm{V_j-\tilde V_j}_{L^2}\le CK\,\occ^{1/2}N^{-1/2}(\log N)^c$,
and the new-view error is bounded by the training error plus an angular
regularity term and the SH/visibility capacity term.
\end{theorem}

\begin{theorem}[A taxonomy of primitive families]\label{thm:u8}
All three primitive families embed in canonical form
$(Q,w,\rho)$; point (mass), wave (Gaussian plate $w=\alpha_n2\pi\sigma_t^2$
in the limit), field (volume compositing), with voxels/triangles as
field discretizations. On surface-dominated scenes,
\begin{equation}
\begin{aligned}
\Rcal_{\mathrm{tri}}&\sim KN^{-1}\ll
\Rcal_{\mathrm{pts/surfel/plates}}\sim KN^{-1/2}(\log N)^c\\
&\ll \Rcal_{\mathrm{iso\ voxels}}\sim KN^{-1/4}.
\end{aligned}
\label{eq:ranking}
\end{equation}
Optimal cells become covariance-ellipsoidal with
density $\rho_A\propto\lambda^{d_{\mathrm{sur}}/(d_{\mathrm{sur}}+2)}$, where $d_{\mathrm{sur}}$ is the
surface (content) dimension, distinct from the phase-space support
dimension $d_{\mathrm{ps}}$ of the Zador law; on surfaces $d_{\mathrm{sur}}{=}2$, so
$\rho_A\propto\lambda^{1/2}$.
\end{theorem}

\paragraph{Proofs and imports for the two theorems above.}
\emph{Descent and stationary balance.} On the audited domain $\Ccal$
and $\Vcal$ are both $L$-smooth, hence so is $\Lcal=\Ccal-\Vcal$; the
descent lemma for the gradient step
$\theta_{t+1}=\theta_t-\eta\nabla\Lcal$ gives
$\Lcal(\theta_{t+1})\le\Lcal(\theta_t)-\eta\bigl(1-\tfrac{L\eta}2\bigr)
\norm{\nabla\Lcal(\theta_t)}^2$,
so the iterates are monotone for $\eta<2/L$ and every accumulation
point satisfies $\nabla\Ccal=\nabla\Vcal$.
\emph{Lloyd.} For atomic families the loss identity proved above
applies with counting measures: the quantization error
$\int\min_i\norm{x-y_i}^2\,\dd\mus$ is minimized alternately over the
assignment (nearest neighbour) and the centres (centroids), each step
an exact block minimization, so the iterates are monotone and
stationary quantizers are fixed points~\cite{graf2000foundations}.
\emph{Classical rates.} The sampling rate $N^{-2/d}$ is Zador's law
\cite{zador1982,graf2000foundations}; the meshing rate $N^{-1}$ on
$C^2$ surfaces is the classical finite-element width
bound~\cite{devore1998}.
\emph{Consistency (structural).} All views couple to the same source
$\mus$ and read the same value field $v$ on $\PS$; view disagreement can
enter only through the direction coordinate of $\PS$ (carried by $v$)
and the angular truncation of the SH basis, which is the stated
angular-regularity tail. The visibility-cell rate and the new-view
bound remain proof sketches: they need the visibility-cell formalism
developed in full, which we do not attempt here.
\emph{Taxonomy.} The embeddings are explicit: point $=$ Dirac source
$\mut=\sum_iw_i\delta_{x_i}$; wave $=$ Gaussian packet with the
footprint kernel of the transport proof; field $=$ the density case of
the same coupling. The ranking: value-only atoms on a surface are
Zador quantization at the content dimension $d_{\mathrm{sur}}=2$,
giving $N^{-1/2}$ (the plate logarithmic loss is the $(\log N)^c$ of
the adaptive-covering theorem); first-order elements on $C^2$ surfaces
give the width bound $N^{-1}$~\cite{devore1998}; isotropic voxels pay
the misclassification band: voxels of side $h$ cut the
$h$-neighbourhood of the surface in ${\sim}h^{-2}$ cells, so a
surface-resolving voxel budget has $N\asymp h^{-2}$ and error
$\sqrt{Kh}\asymp N^{-1/4}$. The ellipsoidal optimum follows the
classical point-density law $\rho_A\propto\lambda^{d/(d+2)}$
\cite{graf2000foundations}; the exact cell shape is Gersho-type, proved
in dimension one and conjectural beyond, and the paper accordingly
reads it as a prediction.

\begin{proof}[Proof of the transport theorem (Theorem~3 of the main paper)]
\textbf{(a) Canonical coupling.} Let $r\colon\PS\to\Pix_v$ be the
measurable ray--pixel assignment and define the finite measure
\[
\piv\bigl(A\times\{p\}\bigr)
:=\int_{A\cap r^{-1}(p)}w(x)\,G\bigl(p-\Pi(x)\bigr)\,
\abs{\det J_\Pi(x)}\,\one_V(x)\,\dd x,
\]
with pixel masses $S(p)=\piv(\PS\times\{p\})$. Both marginals hold by
construction: the pixel marginal is $\nuv(\{p\})=S(p)$, and the source
marginal is the phase-space distribution that axiom~1 of the main paper
assigns to the representation; the atomic case
$\mut=\sum_iw_i\delta_{x_i}$ is the same display over counting measure
and gives the 3DGS instance $\piv(\{x_i\}\times\{p\})\propto w_i(p)$ of
the rendering-coupling definition. Weak continuity in $\theta$: the
weights and footprints are Lipschitz in $\theta$ on the audited domain
and dominated by the bounded transmittance-weighted mass, so dominated
convergence gives $\int f\,\dd\piv(\theta_k)\to\int f\,\dd\piv(\theta)$
for every bounded continuous $f$ along $\theta_k\to\theta$ (weak
convergence; atoms move continuously in $\theta$ by the Lipschitz
regularity).
\textbf{(b) Conditional expectation.} Disintegrating $\piv$ over its
pixel marginal, $\piv(\dd x,p)=S(p)\,\tilde\kappa_v(p,\dd x)$ with
$\tilde\kappa_v(p,\dd x)$ the normalized kernel displayed in the main
paper; the conditional-expectation identity is Bayes' rule for this
disintegration, $I_{v,\theta}(p)=\int \rho\,\dd\tilde\kappa_v(p,\cdot)$.
\textbf{(c) Operator norm, sharp.} Normalize $\nu_v$ to the probability
$\tilde\nu_v=\nu_v/\mu_v$ ($\mu_v=\sum_pS(p)$). For
$P_vf(p)=\int f\,\dd\tilde\kappa_v(p,\cdot)$, Jensen and Fubini give
\[
\begin{aligned}
\norm{P_vf}^2_{L^2(\Pix_v,\nu_v)}
&=\mu_v\int\abs{P_vf}^2\,\dd\tilde\nu_v
\le\mu_v\iint f^2\,\dd\tilde\kappa_v\,\dd\tilde\nu_v\\
&=\mu_v\!\int f^2\,\dd\tilde\mu_\theta,
\end{aligned}
\]
where $\tilde\mu_\theta:=m_v/\mu_v$ is the normalized source marginal of
the coupling ($m_v$ the source marginal of $\piv$; in the conservative
case of (a), $m_v=\mut$ and $\tilde\mu_\theta=\mut/\mut(\PS)$ with
$\mu_v=\mut(\PS)$), and the norm $\norm{\cdot}_{L^2(\PS)}$ of the
theorem is taken with respect to this probability;
$f\equiv1$ attains equality
($P_v1=1_{\Pix_v}$), so $\norm{P_v}=\sqrt{\mu_v}$ is sharp.
\textbf{(d) Blending gap.} On a ray through $x$ write the saturated
compositing weight $C(x)=1-\prod_j(1-\alpha_j)$ and its linearization
$L(x)=\sum_j\alpha_j$. Bonferroni's inequalities give
\[
\abs{C(x)-L(x)}\le\sum_{j<k}\alpha_j\alpha_k
\le\tfrac12\Bigl(\sum_j\alpha_j\Bigr)^{2}\le\tfrac12M(x)^2\varepsilon^2,
\]
and the next order of the product expansion contributes
$O\bigl(M(x)^3\varepsilon^3\bigr)$, which the audited domain's cap on
covering multiplicity absorbs into the stated
$O\bigl(K'M(x)\varepsilon^3\bigr)$; carrying colours multiplies by the
compositing Lipschitz constant $K'$.
\textbf{(e) Chaining.} Decompose
$\nuv-\nu^{\gt}_v=P_v(\mut-\mus)+b_v+t_v$, with $b_v$ the compositing
deviation of (d) and $t_v$ the target-map (anti-aliasing) residual; the
triangle inequality, the sharp bound of (c) transferred through the
target map (Lipschitz constant $\ell_v$, absorbed together with
$\mu_v$ into the displayed $\sqrt{\ell_v}$), and the bound of (d)
yield the three-term chain. \qed
\end{proof}

\section{Notation and terminology}
\label{app:notation}

\begin{table*}[t]
\centering
\small
\begin{tabular}{ll}
\toprule
Symbol & Meaning \\
\midrule
$N$ & primitive budget (number of Gaussians) \\
$N^*$ & floor-knee budget located by the scan \\
$c$ & budget parameter: pixel-footprint threshold $t = c\cdot d/f$ \\
$K$ & capacity budget (per-frame working-set bound) \\
$\theta$ & local redundancy factor of a region \\
$T^*$ & opacity-weighted inertia tensor of a group \\
$W, \mathbf{c}^*$ & weighted volume and centroid of a group \\
$N^*_j \propto K_j s_j / d_j$ & per-tile budget allocation law \\
\midrule
plank & our unit of geometric bookkeeping: an oriented box of half-extents $h=3\sigma$ with opacity weight \\
primitive / splat & one Gaussian of the underlying model (used interchangeably; both are rasterized the same way) \\
anchor & Scaffold-GS-style stored point that decodes to about ten neural Gaussians \\
score export & post-hoc pruning by our contribution ranking \\
\bottomrule
\end{tabular}
\caption{Notation used throughout and the small terminology
glossary (plank vs.\ primitive vs.\ anchor).}
\label{tab:notation}
\end{table*}

\section{Measurement appendix}
\label{app:measure-details}

Records behind Section~6.1 of the main paper: the preregistered protocol and
its four checks, stated in plain narrative form below. The per-band floor
curves, the estimate-bias study, and the downstream uses are archived
with the paper's supplementary artifacts, together with all JSON
certificates.

\subsection{Protocol and records}
\label{app:measure-protocol}

The floor is only useful if it can be verified on real scenes. The floor-measurement
study is preregistered (criteria frozen before data collection) and its
verdicts were frozen before reading the data, a discipline mirrored by
our other measurements (Appendix~\ref{app:protocol}).

\subsection{Protocol}

On each scene: (i) scan budgets (MCMC training at cap
$2\times10^6$ primitives, then opacity-prune to $N$ and fine-tune) to get
the empirical floor curve $R(N)$ ($9$ budgets); (ii) \emph{floor intactness}:
no scan point crosses the fitted floor by more than
$\log_{10}(0.8)$; (iii) published operating points are placed against the
floor, re-anchored to the paper's own baseline at matched capacity ratio,
$5\%$ tolerance $\approx0.22$~dB; (iv) duality certificates
lower-bound achieved error at every (budget, view).

\section{Background and guarantees}\label{app:theory}
We collect the formal background behind the design principles of
Section~4.4 of the main paper: planks and weighted covering, the adaptive
approximation rate, the optimal budget allocation, and the rendering-side
error control. Proofs of the weighted-covering theorem, the adaptive-covering-rate
theorem and its converse, and the derivation of the budget-allocation
corollary are given here; the rendering-side error-control statements
are summarized in the subsection below.

\subsection{From primitives to planks}\label{app:theory-motiv}
A trained model is a set of planks $(\mu_i,R_i,h_i,\alpha_i)$: each
primitive occupies a box of half-extents $h_i=3\sigma_i$ oriented by $R_i$,
carrying weight $\alpha_i$. Dropping primitives means replacing one
covering by a coarser one; the two quantities that decide whether this is
safe are \emph{local covering density} (is this region over-covered, i.e.\
are its planks redundant?) and \emph{approximation error} (how much mass is
lost per removed plank?). The theorems below answer exactly these two
questions, which is why the rendering and training sides of
Sections~6 and~7 of the main paper both reduce to them.

\subsection{Weighted covering}\label{app:theory-covering}
\begin{definition}[Plank]
A \emph{plank} $V$ is a box in $\R^3$ carrying a weight $w_V:V\to[0,1]$
(effective opacity contribution); $|V|$ is its volume. A 3DGS primitive is a
plank: center $\mu$, principal axes from its rotation, half-extents
$h=3\sigma$ (three-sigma convention), weight = opacity.
\end{definition}

\begin{theorem}[Weighted covering]
Let $\T$ be an essentially distinct family of $\delta\times\delta\times1$
tubes with weights $w_T$ satisfying (i) the Katz--Tao density condition
$\Delta_{\max}(\T)\upre1$ (unweighted) and (ii) $\lambda(w)\succap1$.
Then the saturated weighted union volume
$|U|_w:=\int\min\big(1,\sum_Tw_T(x)\big)\dd x$ obeys
$|U|_w\ \gtrsim\ \lambda(w)^{K+1}\,\delta^{\varepsilon}\sum_T\int_Tw_T$,
with $K=K(\varepsilon)$ the constant of the unweighted estimate:
collapsing the family forfeits no mass beyond the $\varepsilon$-loss of
the unweighted theorem and constants depending only on $\varepsilon$.
\end{theorem}

The theorem is the Wang--Zahl~\cite{wz2025} union-volume estimate in its
weighted (continuous-opacity) form: the proof below invokes their
Theorem~1.2, the covering estimate behind the recent resolution of the
three-dimensional Kakeya set conjecture (in the current arXiv version of
\cite{wz2025} the shaded union-volume variant appears as Cor.~1.10). Its
role here is to certify when a group of primitives can be
\emph{collapsed} into fewer representative planks without forfeiting
more than an $\varepsilon$-fraction of volumetric mass. A one-line
corollary caps regional over-packing: for every convex $W$,
$\sum_{T\subset W}\abs{T}\le C\Delta_{\max}(\T)\,\abs{W}$ (sum the count
bound of the density condition over tubes contained in $W$ and multiply
by $\abs{T}$), and this is what keeps the scale-band decomposition of
Section~5.2 of the main paper bounded. Both quantities ($\theta$,
$\Delta_{\max}$) are computable from the factoring algorithm;
the $\theta$ map used by the training-side stop rule
(Section~5.1 of the main paper) is exactly this redundancy signal.
\begin{proof}[Proof (thresholding + union comparison)]
\textbf{Step 1 (thresholding).} Let $s_0:=\lambda(w)/2$ and
$Y(T):=\{x\in T:w_T(x)\ge s_0\}$. By layer-cake,
$\int_Tw_T=\int_0^1|\{w_T>s\}|\dd s$, the low part is
$\le s_0|T|=\frac{\lambda(w)}2|T|$, so the high part is
$\ge\frac{\lambda(w)}2|T|$; hence $|Y(T)|\ge\frac{\lambda(w)}2|T|$.
\textbf{Step 2 (mass comparison).} Pointwise $w_T\le\one_{Y(T)}+s_0$, so
$\int_Tw_T\le|Y(T)|+\frac12\int_Tw_T$, i.e.\ $\int_Tw_T\le2|Y(T)|$.
\textbf{Step 3 (union comparison).} On $U(Y)=\cup_TY(T)$ we have
$\min(1,\sum_Tw_T)\ge s_0$, so $|U|_w\ge s_0|U(Y)|$.
\textbf{Step 4 (invoke the unweighted estimate [WZ25, Thm.~1.2]).}
$(\T,Y)$ satisfies the binary hypotheses with $\lambda_0=\lambda(w)/2$,
so $|U(Y)|\ge\delta^{\varepsilon}\lambda_0^{K}\sum_T|Y(T)|$; chaining
Steps 2--3 and $|Y(T)|\ge\tfrac{\lambda(w)}2|T|$ gives
$|U|_w\ \gtrsim\ \lambda(w)^{K+1}\delta^{\varepsilon}\sum_T\int_Tw_T$.\qed
\end{proof}

\subsection{Adaptive covering rates}\label{app:theory-rate}
\begin{theorem}[Adaptive covering rate]
For $f\in\mathcal{C}_K$ (piecewise $C^2/H^2$ cartoon fields) and a
multiscale dictionary of plates and tubes,\footnote{Terminology: plates
(thin slabs) and tubes are the dictionary atoms of the rate proof; the
bookkeeping unit of the deployment side (Section~5.2 of the main paper) is the
\emph{plank} (an oriented box of half-extents $h{=}3\sigma$,
Appendix~\ref{app:theory-covering}). The three granularities serve
different roles in theory and system and are not identified.}
there is an $N$-atom combination
with $\norm{f-f_N}_{L^2}\le CK\,N^{-1/2}(\log N)^c$.
\end{theorem}

The construction balances cube atoms (smooth part, error $CKN_0^{-2/3}$)
with tangential plates (surface layer, error $CKN_1^{-1/2}$) at the optimal
split $N_0\approx N^{3/4}$, $N_1=N-O(N^{3/4})$. The same construction with
our moment-matched planks achieves the predicted rate on synthetic scenes:
measured log--log slopes $-0.90$ to $-1.09$ vs.\ the theoretical $-1$
(5/5 validation passes; Fig.~\ref{fig:slopes}).

\begin{figure}[h]
\centering
\includegraphics[width=0.6\linewidth]{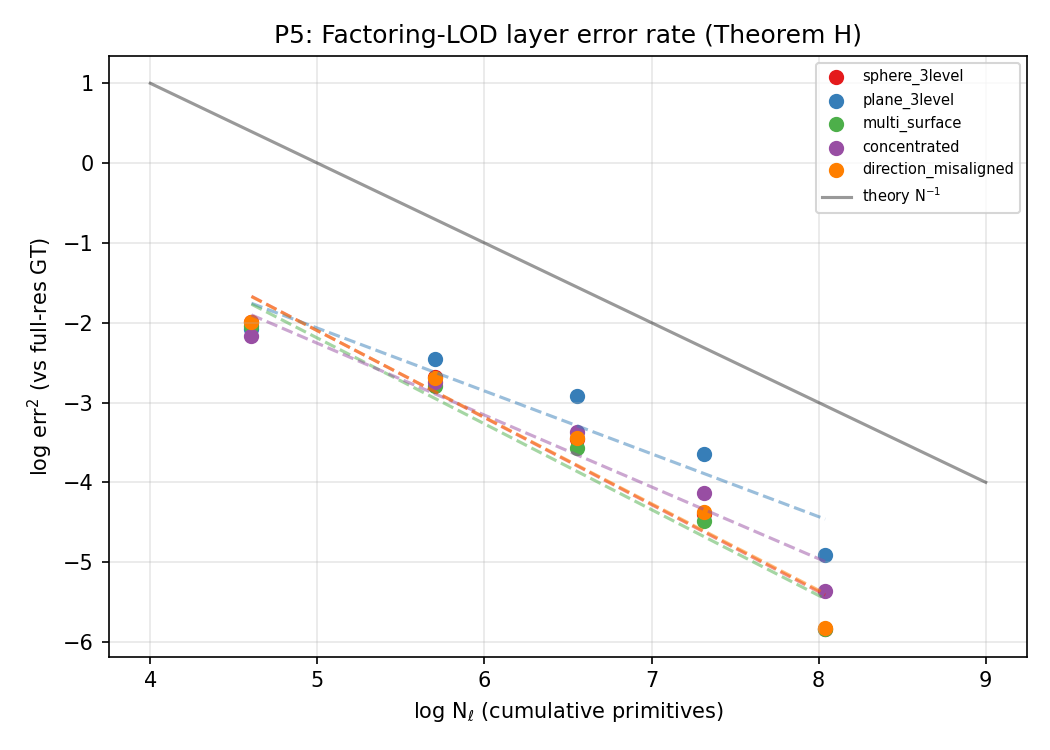}
\caption{Layer error on five synthetic scenes (sphere, plane,
multi-surface, concentrated, direction-misaligned): measured log--log
slopes $-0.90$ to $-1.09$ vs.\ the theoretical $-1$
(err$^2\sim N^{-1}$), 5/5 validation passes.}
\label{fig:slopes}
\end{figure}

\begin{lemma}[Scale truncation]\label{lem:lb0}
For $f\in\mathcal{C}_K$ (piecewise $C^2$) and target precision
$\varepsilon$, atoms of volume $v<v_{\min}:=(\varepsilon/C_6K)^2$ can be
discarded from any $\varepsilon$-approximation with loss at most
$\varepsilon/2$.
\end{lemma}

This is the precise permission our scheduler uses: atoms whose projected
footprint is below the pixel scale can be dropped or collapsed with
bounded loss, the license for the
pixel-footprint threshold $t=c\cdot d/f$ of Section~5.3 of the main paper, and
equivalently the license behind the rate-based stopping criterion of the
training side (once pruning stops improving the rate, the budget is
exhausted). Status: a standard nonlinear-width result~\cite{devore1998};
we claim no new proof. (The naive per-atom Poincar\'e oscillation
argument fails under the $H^1$-only class: if $f$ itself is a fine
atom, truncation is invalid; the correct route uses the
piecewise-$C^2$ structure, the packing-element decomposition, and
essential distinctness.)

\begin{proof}[Proof of Theorem~H (upper bound)]
For $f\in\mathcal{C}_K$ (piecewise $C^2/H^2$):
\textbf{(a) Smooth part.} Tile $B_1$ with cubes of diameter $s$
($N_0=s^{-3}$); the second-order Poincar\'e inequality per cube gives
error $CKN_0^{-2/3}$. \textbf{(b) Surface layer.} Cover the
$s$-neighbourhood of the $C^2$ surface by $N_1=CK/s$ tangential plates;
the misclassification volume gives error $\le C\sqrt{Ks}=CKN_1^{-1/2}$.
\textbf{(c) Budget.} Balance $N_0^{-2/3}$ with $N_1^{-1/2}$:
$N_1=N_0^{4/3}$, i.e.\ $N_0\approx N^{3/4}$, $N_1=N-O(N^{3/4})$; total
error $\le CKN^{-1/2}(\log N)^c$. \textbf{(d)} Direction mismatch by
$\beta$ degrades scales $s<\beta^2$. Our tree realizes the construction
with moment-matched planks; the $(\log N)^c$ factor is absorbed by the
banding (each band's members have comparable scale). \qed
\end{proof}

\begin{theorem}[Rate lower bound]\label{thm:j}
Let $\mathbb{D}$ be an essentially distinct multiscale plank dictionary
in $B_1$. There exist absolute constants $c_0>0$ and $N_0\ge2$ such
that for every $N\ge N_0$ there is $f^*\in\mathcal{C}_K$ with
$\norm{f^*-f_N}_{L^2}\ge c_0\,K\,(N\log_2N)^{-1/2}$ for every $N$-atom
combination $f_N$ from $\mathbb{D}$.
\end{theorem}

\begin{proof}[Proof (capacity argument)]
\textbf{Step 1 (packing).} Take a cartoon field whose jump set is a
$C^2$ sheet of area ${\sim}K$ and fragment it into ${\sim}K/s^2$
patches of diameter $s$; for a sign vector $\sigma$ shift patch $Q$ by
$\sigma_Qs$ along its normal. Every $f_\sigma$ remains in
$\mathcal{C}_K$ (the fragmentation walls add area ${\sim}K$), and two
configurations differing on a constant fraction of the patches have
squared distance ${\sim}(K/s^2)\,s^3=Ks$; choosing
$s\asymp\varepsilon^2/K$ (so $\sqrt{Ks}\ge2\varepsilon$) gives
$\log M_{\mathrm{pack}}\ge c_5\,K/s^2\ge c_5(K/\varepsilon)^2$.
\textbf{Step 2 (covering).} By Lemma~\ref{lem:lb0} only atoms of volume
$\ge v_{\min}$ matter; essential distinctness caps the multiplicity at
$C_0$ per (direction cap $\times$ scale layer), so the dictionary has
$\le C_8(K/\varepsilon)^2\log(1/\delta)$ relevant atoms; quantizing the
$N$ coefficients at step $\eta=\varepsilon/(2B^{1/2}\sqrt{N})$ gives an
$\varepsilon$-net of all $N$-term combinations with
$\log M_{\mathrm{cov}}\le C_{11}N\log(NK/\varepsilon)$ (the truncation
makes $\log(1/\delta)$ absorbable).
\textbf{Step 3 (capacity collision).} Covering forces
$M_{\mathrm{cov}}\ge M_{\mathrm{pack}}$, i.e.
$C_{11}N\log(NK/\varepsilon)\ge c_5(K/\varepsilon)^2$; setting
$\varepsilon^*=c_{12}K(N\log_2N)^{-1/2}$ with
$c_{12}=\sqrt{c_5/2C_{11}}$ contradicts this for all $N\ge N_0$
(absolute), so some $f^*$ admits no $N$-term $\varepsilon^*$-approximation.
\textbf{Step 4 (applicability).} Outside the balanced regime the claimed
bound is smaller than the diameter of $\mathcal{C}_K$ and the statement
holds trivially; $c_0$ is the minimum over the two regimes. \qed
\end{proof}

The upper and lower bounds differ only in the power of $\log N$ (dyadic
layer bookkeeping vs.\ coefficient quantization in the capacity
collision); the exact exponent is not pinned down, and no functional
form used in this paper depends on it.

\subsection{Optimal budget allocation}\label{app:theory-budget}
\begin{corollary}[Budget allocation]
Minimizing $C_1KN^{-1/2}+C_2K'L^{-p}$ (geometric $+$ appearance error)
under the parameter budget $M\asymp N(L+1)^2$ yields
$N^*\asymp M^{p/(p+1)}$ and $L^*\asymp M^{1/(2(p+1))}$.
\end{corollary}

Specialized to a hierarchy with scale bands $j=0,\dots,J$ whose members
have characteristic scale $s_j=2^j s_0$, the allocation assigns band budgets
$m_j\propto 2^{-j\gamma}$ for a scene-dependent $\gamma$: coarser bands
(which cover more volume per plank) receive exponentially more budget.
Our scheduler instantiates this allocation \emph{per view} through the
pixel-footprint rule (Section~5.3 of the main paper): the budget of a band is
exactly the set of its members whose projected scale exceeds the
resolvability threshold.
Minimize $E(N,L)=C_1KN^{-1/2}+C_2K'L^{-p}$ under $NL^2\asymp M$
($L\gg1$). Substituting $N=M/L^2$: $E(L)=C_1KM^{-1/2}L+C_2K'L^{-p}$;
the first-order condition gives
$L^*=(pC_2K'/C_1K)^{1/(p+1)}M^{1/(2(p+1))}$ and
$N^*=M/(L^*)^2$, and $E''>0$ confirms the unique minimum.

\subsection{Rendering-side error control}\label{app:theory-render}
Replacing primitives by equivalent planks changes the rendered image; the
end-to-end certificate bounds this change in the audited domain:
per-view image error is controlled by the field error transferred through
the rendering operator term plus a model-operator gap classified by a
saturation-front decomposition. In practice we do not rely on the
(loose) worst-case constants: the tree additionally stores per-cluster
\emph{residual statistics}, geometric spread $\sigma_{\mathrm{geo}}$ and
SH-residual ratio $\rho_{\mathrm{sh}}$, which act as an empirical error
proxy for the switch decision (Section~5.3 of the main paper). We verify on
real scenes that the certificate remains valid (never violated in
measurements) but loose; this is stated as a limitation
(Section~7 of the main paper).

\section{Experimental protocol}\label{app:protocol}

\textbf{Implementation keys (main paper \S5.1--5.2).} The training-side
scan and capped training live in a single script; its configuration keys,
for reproducibility: BCT\_INIT\_FRAC (initial primitive count as
a fraction of the cap $N^*$), BCT\_REFINE\_STOP\_FRAC (refine-stop
fraction, both additionally scaled by the measured redundancy
$\theta_{\mathrm{waste}}$), prune\_rate$=0.25$ (per-step pruning
rate inside $\theta$-flagged regions), and agg\_cap$=256$
(maximum equivalent planks per (leaf, band) bucket).

\textbf{Data sources.} lego/chair/ship: Blender NeRF-Synthetic;
garden: Mip-NeRF 360 (images\_4); MatrixCity Aerial: official
CityGaussianV2 pretrained checkpoint
(MatrixCity\_Aerial.zip, 8,558,267 Gaussian2D primitives,
sh degree 2) and the block-formatted dataset (block\_1 + block\_2,
1,226 frames with images).

\textbf{Camera chain (MatrixCity).} Poses follow the official pipeline:
rot\_mat in the block transforms.json is
camera-to-world; apply the OpenGL$\to$COLMAP axis flip
($c2w[:3,1:3]\times=-1$), then $w2c=\mathrm{inv}(c2w)$. Resolution:
1.2$\times$ downsampling to 1600$\times$900; $f_x=(W/2)/\tan(\mathrm{fov}/2)$
with $\mathrm{fov}=45^\circ$ $\Rightarrow$ $f_x\approx1931.37$,
$c_x=800,c_y=450$. Test views: every 8th frame of the merged sequence
(154 views; 24-view uniform subset for paired comparisons). GPS rendering
uses the vertical fov $2\arctan(H/2f_x)=26.2313^\circ$ (its
{-}{-}fovy is \emph{vertical}). The block transforms.json
contains duplicate frame\_index values (255 duplicates); all
cameras are taken from the single export
(cameras\_test.json) to guarantee GT--pose pairing.

\textbf{2D$\to$3D conversion.} Gaussian2D surfels ($\sigma\in\R^2$) are
converted to 3D planks with thickness
$h_3=0.02\cdot\min(h_1,h_2)$; \emph{every} comparison cell (full, LOD,
GPS, random) uses the same conversion and the same renderer
(gsplat 1.5.3, or GPS with SH zero-padded to degree 3 and scale/opacity
in log/logit domain).

\textbf{Metrics and hardware.} PSNR against GT photographs (LANCZOS
downsampled to the protocol resolution); paired subsets report median
alongside mean. All GPU work an A100 GPU-80GB (GPU 2), CUDA 12.6,
PyTorch 2.x, gsplat 1.5.3; GPS commit 32b1315 with headless EGL patches
for server rendering. Memory: model data = fp32 bytes per primitive
(152 B gsplat sh2; 236 B GPS sh3); process peaks via
nvidia-smi {-}{-}query-compute-apps (GPS full: 2364 MiB) and
torch.cuda.max\_memory\_allocated (LOD cells).

\textbf{Random-subsampling control.} Global random subset per budget
(seed 0, primitive counts matched to the LOD tier averages), exported
through the identical PLY path and rendered by GPS under the identical
camera set.

\textbf{City-scale pruning and real-time protocol
(Block\_all / Block\_A).} Block\_all is a $19{,}932{,}287$-primitive
model of the merged block\_1$+$block\_2 region; Block\_A is a
$7.94$M-primitive block model. $\theta$-pruning removes primitives by
ascending opacity inside $\theta$-flagged (over-covered) regions to a
target count (offline, numpy; identical signal to the training-side pruning
step). Self-PSNR references are rendered with the same CUDA rasterizer
on the A100 (five spherical views at distances 25/12/5 around the scene
center, $\mathrm{fov}_x{=}40^\circ$, $800^2$, white background, fp32
before RGBA quantization); the quality reference scan
(opacity-ranked subsets vs.\ the full model) uses gsplat on both sides.
The C++ rasterizer and gsplat disagree by a roughly constant
$\approx$11.6 dB on identical inputs (rasterization implementation
differences, independent of primitive count), which is why every
quality comparison keeps the same renderer on both sides.
Pruning baselines (LightGaussian, RAP) are applied to the same Block\_A
model at matched primitive counts and evaluated with gsplat (DC-only,
$1000^2$, 5 held-out views) against real photographs. The
contribution-scored export uses the same protocol: renders are produced
by a verbatim port of the baseline rendering script (validated by
reproducing the Vanilla anchor exactly: $13.85$/$0.519$/$0.553$), with
SSIM/LPIPS from the identical evaluation code. The
scoring pass itself is training-view-only and touches no test view.

\textbf{Realtime protocol.} Two measurement setups. (a)
\emph{Interactive city-scale protocol}: RTX 4080 SUPER, native
$1920\times1080$, full SH, headless. The official Block\_all test
trajectory is the CityGaussianV2 evaluation split (207 views;
test.txt of block\_all\_colmap\_eval, camera chain
transform\_matrix~$\cdot\,\mathrm{diag}(1,1,-1)$); orbit probes
use a far view at distance 25. The v3 retrained model ($6.89$M
primitives, $1.7$~GB) is measured on this trajectory (quality on the
reference gsplat path, speed on both rasterizers), while the
capacity-ladder and $800^2$ rows use the earlier v2 tree ($7.07$M) under
the identical protocol. Timing: our CUDA rasterizer uses stage-summed
CUDA timings (60-frame window, first 10 dropped); the gsplat reference
uses torch-synchronized per-frame timing (first 10 dropped); Octree-GS
is re-rendered on the same 207 views from its official checkpoint (its
loader natively downsamples to $1600\times900$) with CUDA-synchronized
per-view averages. (b) \emph{Commodity 8~GB reference}: RTX 3060 Ti,
$800^2$, DC-only compact path, 30-frame headless average, on the
original $19.93$M Block\_all model. The rasterizer is derived from
an existing splatting-rasterizer codebase, while the LOD layer (.factree format,
GPU tier selection, gather kernels) was implemented for this work.

\textbf{Code and data.} The tree builder, schedulers (Python reference
and CUDA), and evaluation scripts will be released upon publication; all
theorems used in the main text are proved in Appendices~\ref{app:sec-theory}
and~\ref{app:theory}.

\section{Seed robustness of Budget-Centered Training}
\label{app:seed-robustness}

All Budget-Centered Training rows of Table~1 of the main paper are
reported as mean$\pm$std over three seeds (s0/s1/s2) under a single,
re-measured protocol (CityGaussian pipeline, MCMC cap at the budget,
SfM init of 100k atoms, 30k iterations). Per-scene values:

\begin{table}[h]
\centering
\small
\begin{tabular}{lrrrr}
\toprule
scene@budget & s0 & s1 & s2 & mean$\pm$std \\
\midrule
counter@536k & 28.93 & 28.98 & 28.96 & 28.95$\pm$0.02 \\
counter@560k & 28.96 & 28.99 & 29.00 & 28.98$\pm$0.02 \\
garden@1650k & 27.12 & 27.12 & 27.21 & 27.15$\pm$0.04 \\
garden@2209k & 27.26 & 27.18 & 27.32 & 27.25$\pm$0.06 \\
kitchen@1M & 31.84 & 31.99 & 32.08 & 31.97$\pm$0.10 \\
kitchen@1131k & 32.28 & 32.10 & 32.16 & 32.18$\pm$0.07 \\
flowers@1330k & 20.56 & 20.61 & 20.58 & 20.58$\pm$0.02 \\
flowers@1525k & 20.62 & 20.59 & 20.60 & 20.60$\pm$0.01 \\
room@529k & 31.89 & 32.00 & 31.94 & 31.94$\pm$0.04 \\
bonsai@601k & 32.13 & 32.10 & 32.14 & 32.12$\pm$0.02 \\
bonsai@640k & 32.18 & 32.11 & 32.18 & 32.15$\pm$0.03 \\
\bottomrule
\end{tabular}
\caption{Three-seed BCT results (test PSNR, dB, val split of the
training pipeline). Seed spread is $\pm$0.01--0.10~dB across all
eleven cells.}
\label{tab:seed3}
\end{table}

Two protocol notes. \textbf{(i) Configuration drift.} Rows measured
under an earlier configuration of the same pipeline (published in the
main table before this re-measurement) read 0.4--1.7~dB lower on
kitchen and room; the current-protocol values above supersede them,
and the main-paper table uses exactly these numbers.
\textbf{(ii) flowers data.} flowers uses the COLMAP-format copy that
the floor study used (173 images, COLMAP split); the released
360\_v2.zip does not contain a flowers directory, so the two data
sources are not interchangeable and we state the source explicitly.

\section{Run-to-run stability of capped MCMC}
\label{app:mcmc-stability}

During the seed-completion campaign we observed occasional training
collapses under capped MCMC (final test PSNR $\approx$8--15~dB instead
of $\approx$30--32~dB, with the loss converged normally): 2 of 12
runs in one batch (garden@1650k s2, kitchen@1131k s1), consistent with
the run-to-run nondeterminism already disclosed for room@529k in the
main table (three-run median $\pm$0.2~dB). Same-seed retries recover
the expected quality (kitchen@1131k s1: 14.67 on retry). Under the
formal pipeline, 0 of 24 runs collapsed. Policy: any cell whose three
seeds disagree by more than 1~dB is re-run once and the median is
reported; both policies mirror the main-paper footnote for room@529k.

\section{Tree hyper-parameter sensitivity}
\label{app:tree-abl}

The factoring tree's two defaults (aggregate cap 256, leaf capacity
4096) were validated by a $3\times3$ sweep on garden
(agg\_cap $\in$ \{128, 256, 512\} $\times$ leaf
$\in$ \{2048, 4096, 8192\}): the full-model sanity PSNR across all
nine combinations stays within 20.43--20.66~dB (spread 0.23~dB), with
the best thickness fraction identical (0.02) in every cell. The
scheduler defaults are therefore insensitive at the model level.

\section{Timing of BCT vs.\ the traditional route}
\label{app:timing}

Measured on room@529k (one A100, neighbor jobs present on the same
device; ratios are same-card): BCT is a single run of
$78\pm7$~min over three seeds. The traditional route for the same
budget -- train at 1M (125~min, measured), prune to 529k, fine-tune --
totals $\approx$175~min, i.e.\ $2.2\times$ BCT's wall-clock, while
reaching lower quality ($\approx$30.5 vs.\ 31.9~dB under the
re-measured protocol). The main-paper phrase ``one 78-minute run
instead of a 175-minute train--prune--fine-tune pipeline'' refers to
this same-card measurement.

\section{Scheduler enhancements and export}
\label{app:sched-enhance}

This section collects the deployment-side enhancements of the
pixel-footprint scheduler (Section~5.3 of the main paper): adaptive
selection, the renderer-agnostic export path, and render-aware
scoring.

\textbf{Engineering notes (moved from the main text).} Construction of
the tree on $8.56$M primitives takes $1.5$~s with a $19.4$~GiB peak on
an A100 GPU. The compact binary stores raw
primitives in Morton order, per-band equivalent planks, and SH first
moments of all degrees; per frame, the coverage-preserving prefix cut is
traversed on GPU and the selected prefix is gathered by a per-primitive
kernel directly into the rasterizer buffers (fused with SH evaluation
and opacity compensation), with no CPU-side collection or
host-to-device upload: $0.8$~ms preprocess including selection$+$gather
on the $7.07$M tree at $1920\times1080$. The GPU selection reproduces
the offline Python reference exactly (identical per-view tier counts on
5/5 views at $c{=}1$ and 10/10 at $c{=}8$). The DC-only compact path
stores $3$ floats per primitive and selects $13.1$M primitives at
$3.9$~GB on 8~GB GPUs. On the official 207-view trajectory at capacity
$K{=}2$M the per-frame selected count changes by a median of $1.1\%$
between consecutive views; the large-count outliers coincide with
view-direction reversals, where the leaves entering or leaving the
selection are back-facing for the new view and are never rasterized.

The serial experiment of Section~6.2 of the main paper (room, 34 held-out views, LOD
render protocol) is summarized below; per-view data and the script are
archived with the paper artifacts.

\begin{table*}[t]
\centering
\small
\begin{tabular}{llrrr}
\toprule
root model & tier & gt-PSNR (dB) & rendered primitives \\
\midrule
\multirow{6}{*}{\shortstack[l]{full model\\(1.66M primitives)}}
 & full render & 34.94 & 1{,}657{,}466 \\
 & LOD $c{=}1$ & 34.89 & 1{,}654{,}484 \\
 & LOD $c{=}2$ & 33.29 & 1{,}542{,}467 \\
 & LOD $c{=}4$ & 30.93 & 1{,}354{,}716 \\
 & LOD $c{=}8$ & 26.92 & 1{,}003{,}828 \\
 & LOD $c{=}16$ & 22.83 & 524{,}330 \\
\midrule
\multirow{6}{*}{\shortstack[l]{budget-centered model\\(529k primitives)}}
 & full render & 34.35 & 529{,}000 \\
 & LOD $c{=}1$ & 34.19 & 517{,}599 \\
 & LOD $c{=}2$ & 33.11 & 455{,}635 \\
 & LOD $c{=}4$ & 30.04 & 350{,}328 \\
 & LOD $c{=}8$ & 25.67 & 209{,}254 \\
\bottomrule
\end{tabular}
\caption{Serial composition on room: both roots deployed through the
same tree and scheduler at matched rendered budgets. The
budget-centered root passes through the tree essentially unchanged
($-0.16$~dB at its own budget), while the full root collapses by
$12.1$~dB at the same count.}
\label{tab:serial-room}
\end{table*}

For adaptive (scene-content-aware) selection we additionally store per
cluster two residuals: the geometric spread
$\sigma_{\mathrm{geo}}=\sqrt{\operatorname{tr}(T^*)/W}$ (world units,
projected to pixels by $\sigma\cdot f/d$) and the SH residual ratio
$\rho_{\mathrm{sh}}$ (cluster-internal SH variance vs.\ global RMS). These
proxy the switch error analyzed in the error-control section above; on
garden they improve the fixed-threshold LOD by $+0.49$~dB at a
comparable primitive count (1.05$\times$; Table~2 of the main paper).

The scheduler emits a standard primitive stream
(means, quaternions, scales, opacities, SH), consumable by any splatting
rasterizer (gsplat, CityGaussian, etc.). For cross-renderer evaluation we
export each tier selection to standard 3DGS PLY files (SH zero-padded to
degree 3, scale/opacity in log/logit domain) and render them with the GPS
renderer~\cite{rijsdijk2026gps} (Section~5.3 of the main paper).

A purely geometric redundancy signal
has a blind spot: a primitive can be geometrically redundant yet
visually load-bearing (it occludes its neighbors, or carries a colour no
neighbor reproduces), or geometrically isolated yet invisible (fully
occluded in every training view). Aggressive pruning exposes exactly
this gap: a geometry-only score removes primitives that a
contribution-aware score would keep (Table~S7); the
underlying occlusion term is an open item in the error control of
the supplementary material.

For pure primitive-count budgets,
the same export path supports a render-aware ranking. We score every
primitive with the LightGaussian
significance~\cite{fan2024light} (the accumulated per-pixel
$T\alpha$ over all training views, where $T$ is the front-to-back
transmittance) via a one-pass forward scan over the training cameras
($234$~s for 4565 views at $7.94$M primitives on an RTX 4080 SUPER). Our
CUDA accumulation replicates the reference rasterizer's per-pixel
semantics exactly (validated by the telescoping identity
$\sum_i s_i=\sum_p(1-T_{\mathrm{final}})$, matched to $10^{-5}$ relative
on random scenes). Any budget point is then a top-$N$ re-export of the
same trained model. On Block\_A, $24.3\%$ of the trained primitives
score exactly zero, never visible from any training view, and the
resulting budget axis stays at Vanilla quality down to $25\%$ of the
primitives (Section~6.5 of the main paper).

\section{Extended experimental results}
\label{app:ext-results}

The tables and figures in this section extend the experimental report of Section~6 of the main paper: outdoor BCT runs, the lego ceiling study, multi-scene generalization, the training-coupled-LOD comparison, Block\_A matched-budget pruning, the published-frontier and universal-floor figures, and the remaining perception/orbit/visual comparisons.

\begin{figure}[h]
\centering
\includegraphics[width=0.85\linewidth]{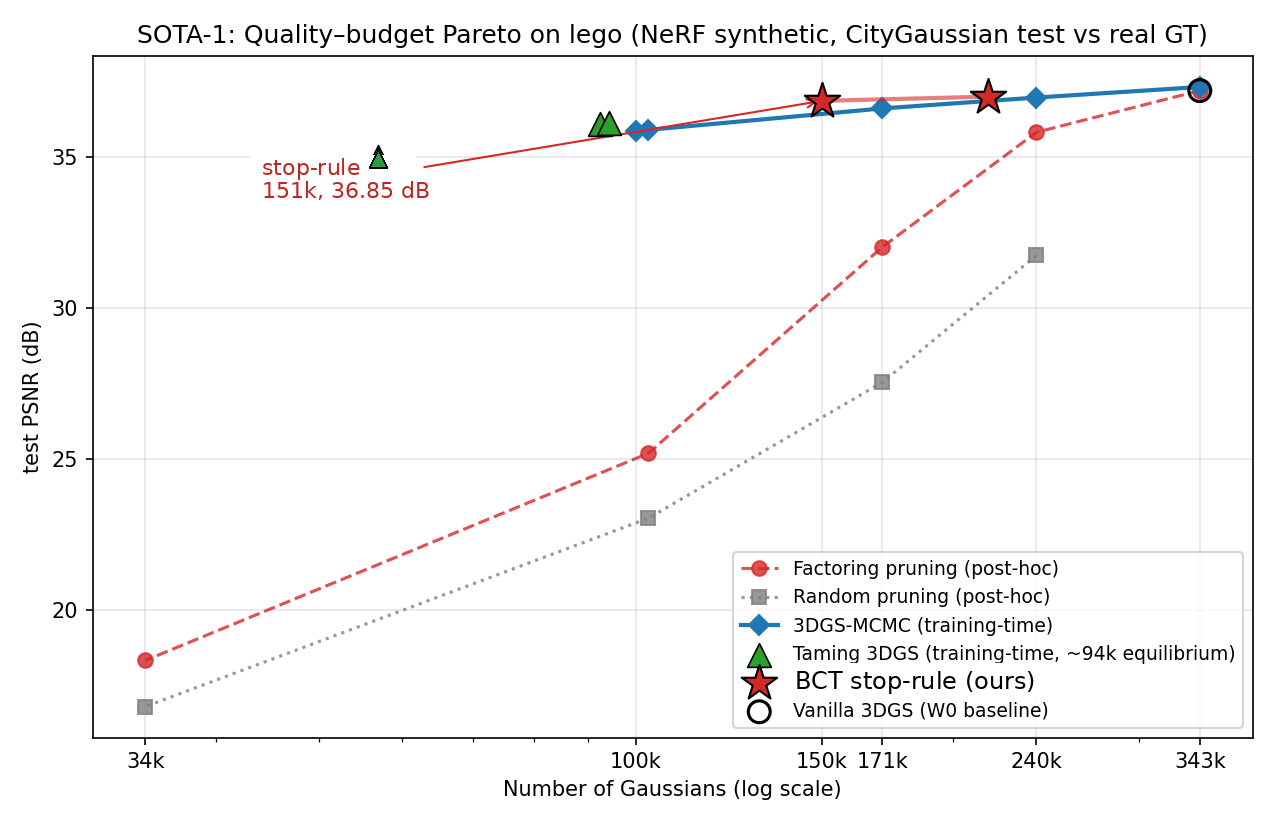}
\caption{lego primitive-budget Pareto: the BCT stop rule dominates the
MCMC budget curve; Vanilla is a single point at 339k (the W0 baseline
of our training-dynamics study).}
\label{fig:pareto_lego}
\end{figure}

\begin{figure*}[t]
\centering
\includegraphics[width=0.99\textwidth]{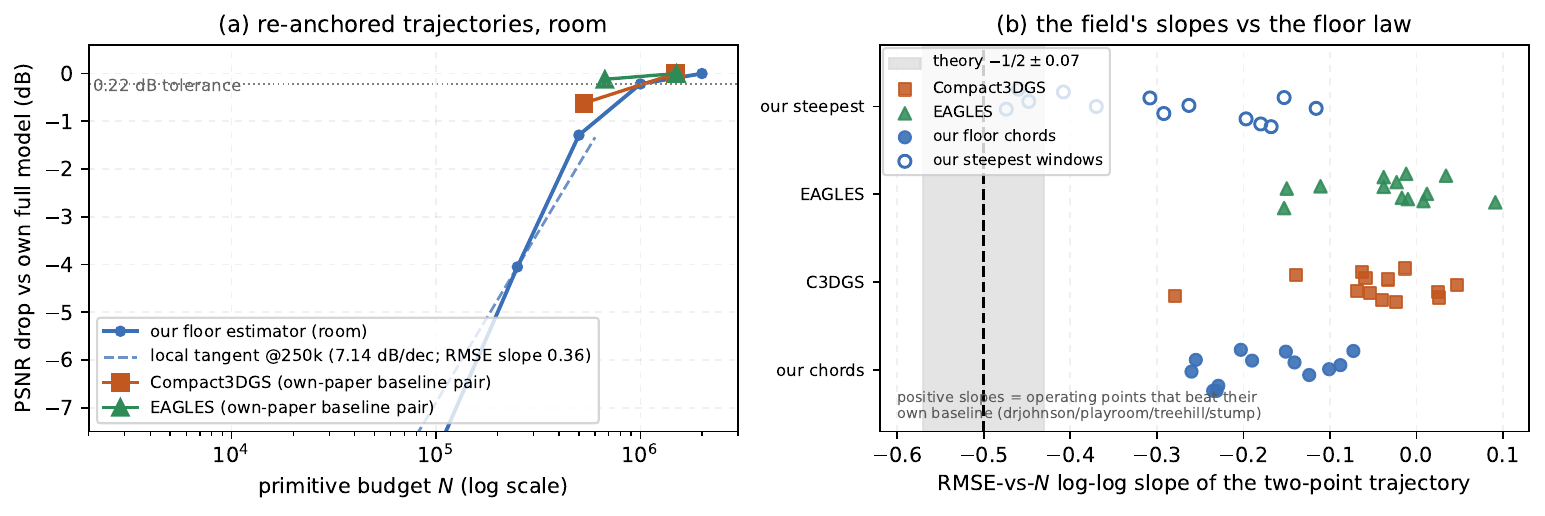}
\caption{The published frontier vs.\ the floor: (a) room, re-anchored to
each paper's own full model, against the published same-paper pairs; (b)
the RMSE-vs-$N$ slope of every surveyed two-point trajectory: the
surveyed compressors' slopes cluster near zero (mean $-0.04$), an order of magnitude
away from the floor law.}
\label{fig:frontier}
\end{figure*}

\begin{figure*}[t]
\centering
\includegraphics[width=0.9\textwidth]{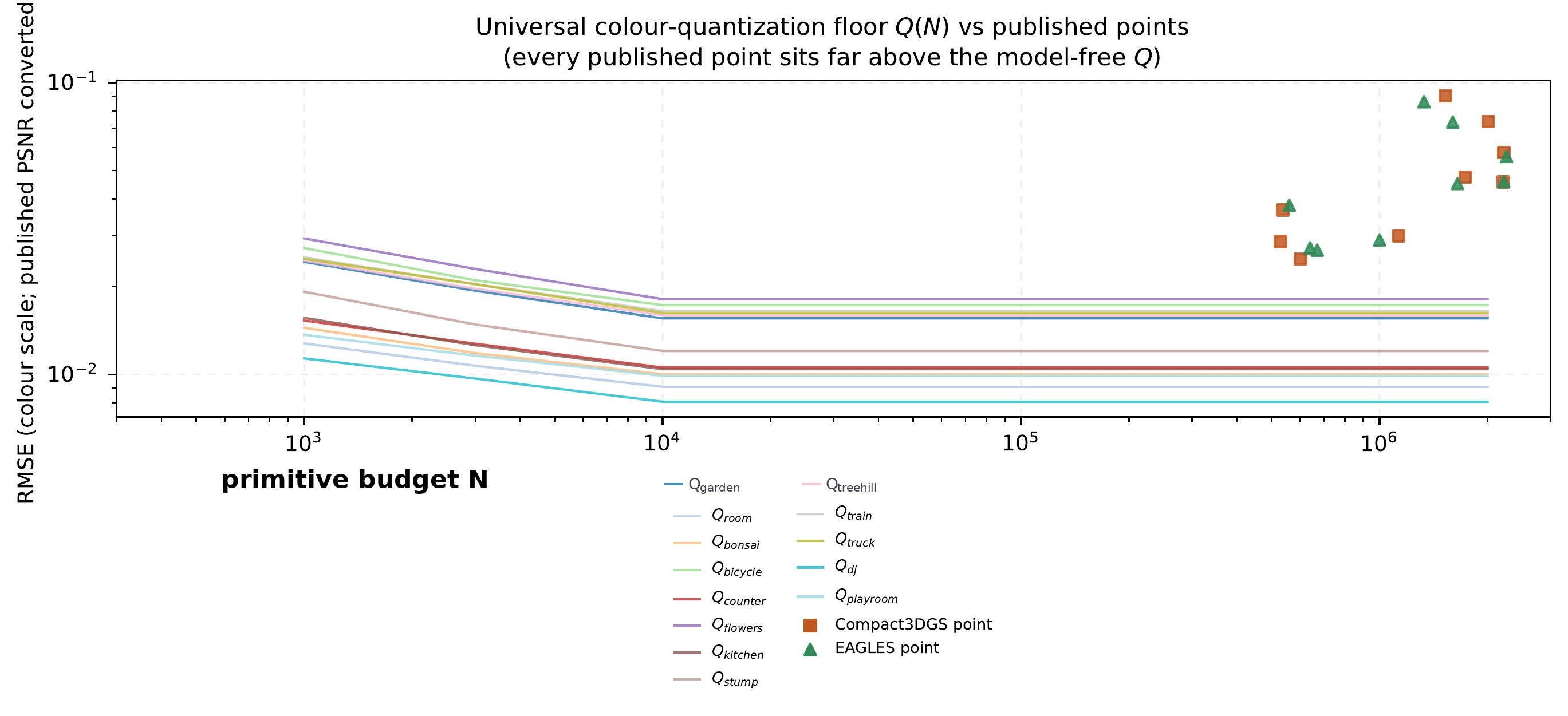}
\caption{The universal (model-free) floor: image-only colour
quantization vs.\ every surveyed operating point; all lie above the
band.}
\label{fig:universal}
\end{figure*}

\begin{figure*}[t]
\centering
\includegraphics[width=\linewidth]{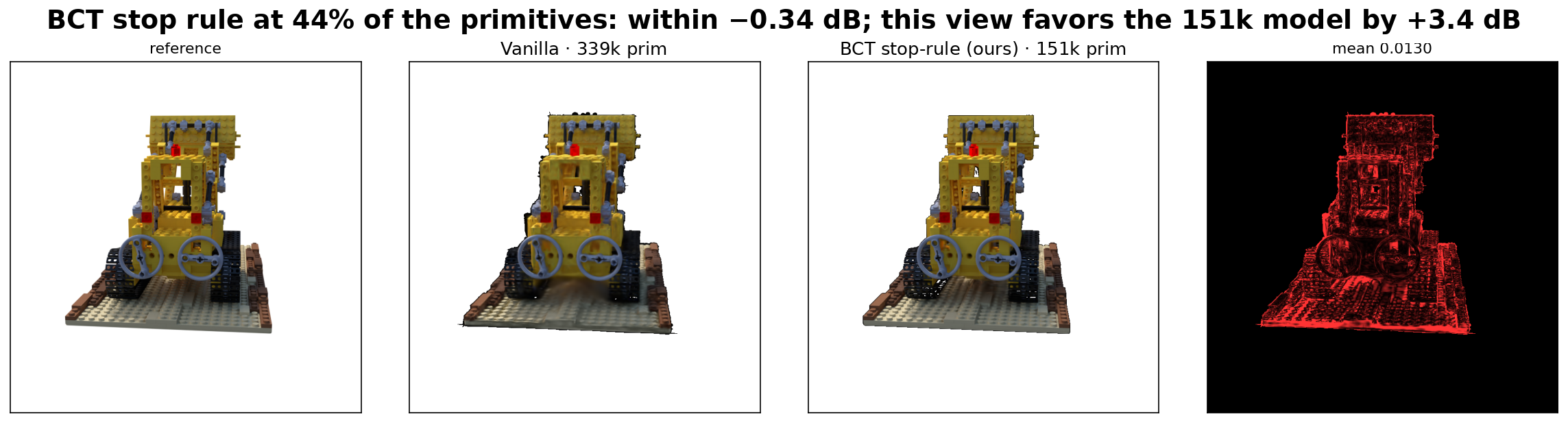}
\caption{lego perception check: reference, Vanilla, stop-rule 151k,
and $|\mathrm{diff}|\times10$ between the two renders. At this single
held-out view the 151k model is $+3.4$~dB over Vanilla (the full
test-set average is $-0.34$~dB, Table~\ref{tab:lego}); residuals
concentrate on structure edges.}
\label{fig:lego_perceptual}
\end{figure*}

\begin{figure*}[t]
\centering
\includegraphics[width=\linewidth]{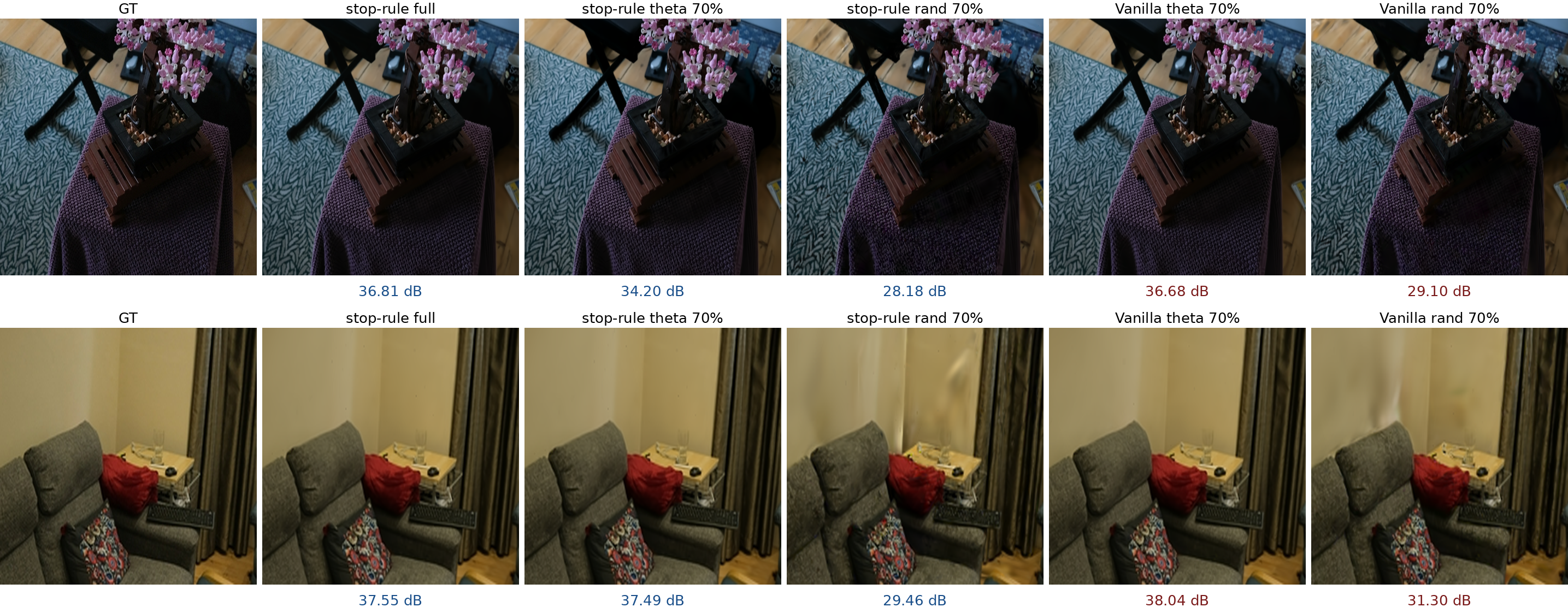}
\caption{Mip-NeRF 360 indoor scenes, held-out view, 70\% budget
(gsplat re-render of the same checkpoints; per-view PSNR beneath each
tile). Top: bonsai; bottom: room. Reference photo; the stop-rule
model in full, pruned to 70\% by $\theta$ guidance vs.\ random; and the
Vanilla model pruned the same two ways. The visible gap (foliage on
bonsai, sofa on room) is consistent with the test-set margins over
random at matched budget ($+3.0$~dB on bonsai, $+6.1$~dB on room; the
per-view margins on this view are larger).}
\label{fig:theta_prune}
\end{figure*}

\begin{figure*}[t]
\centering
\includegraphics[width=\linewidth]{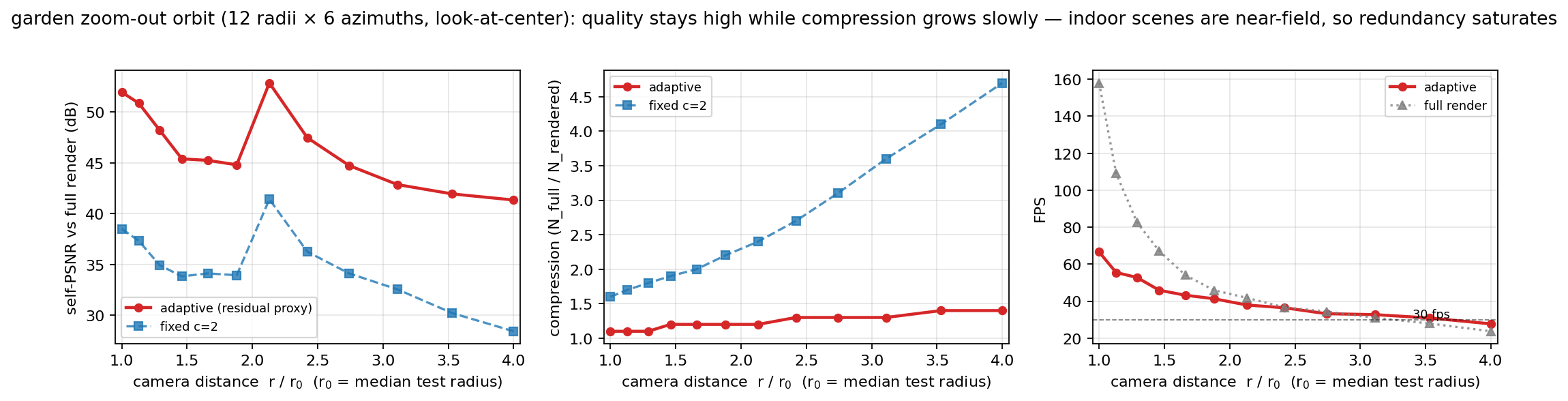}
\caption{garden zoom-out orbit ($12$ radii $\times$ $6$ azimuths,
look-at-center; self-PSNR vs.\ the full render at the same camera).
Adaptive selection stays $41$--$53$~dB while its compression saturates at
$\sim$1.4$\times$; the fixed $c{=}2$ schedule compresses more (up to
4.7$\times$) but loses fidelity. Indoor scenes are near-field: sub-pixel
redundancy exhausts quickly, matching the theory and the table above;
the right panel shows the resulting frame rate.}
\label{fig:garden_orbit}
\end{figure*}

\begin{figure*}[t]
\centering
\includegraphics[width=\linewidth]{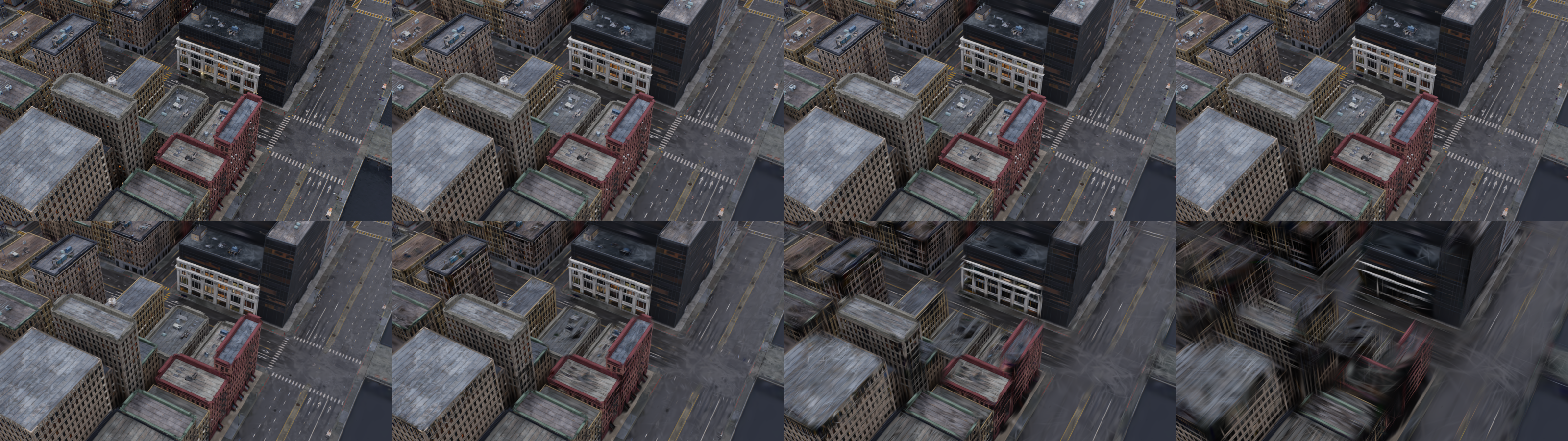}
\caption{MatrixCity Aerial, two held-out views (top/bottom rows), left
to right: reference photo, full render,
LOD $c{=}16$ (23\% of primitives, $-4.7$ dB), and LOD $c{=}32$ (8.9\%,
$-7.6$ dB); the budget parameter removes detail gradually instead of
dropping primitives blindly.}
\label{fig:mc_visual}
\end{figure*}

\begin{figure*}[t]
\centering
\includegraphics[width=\linewidth]{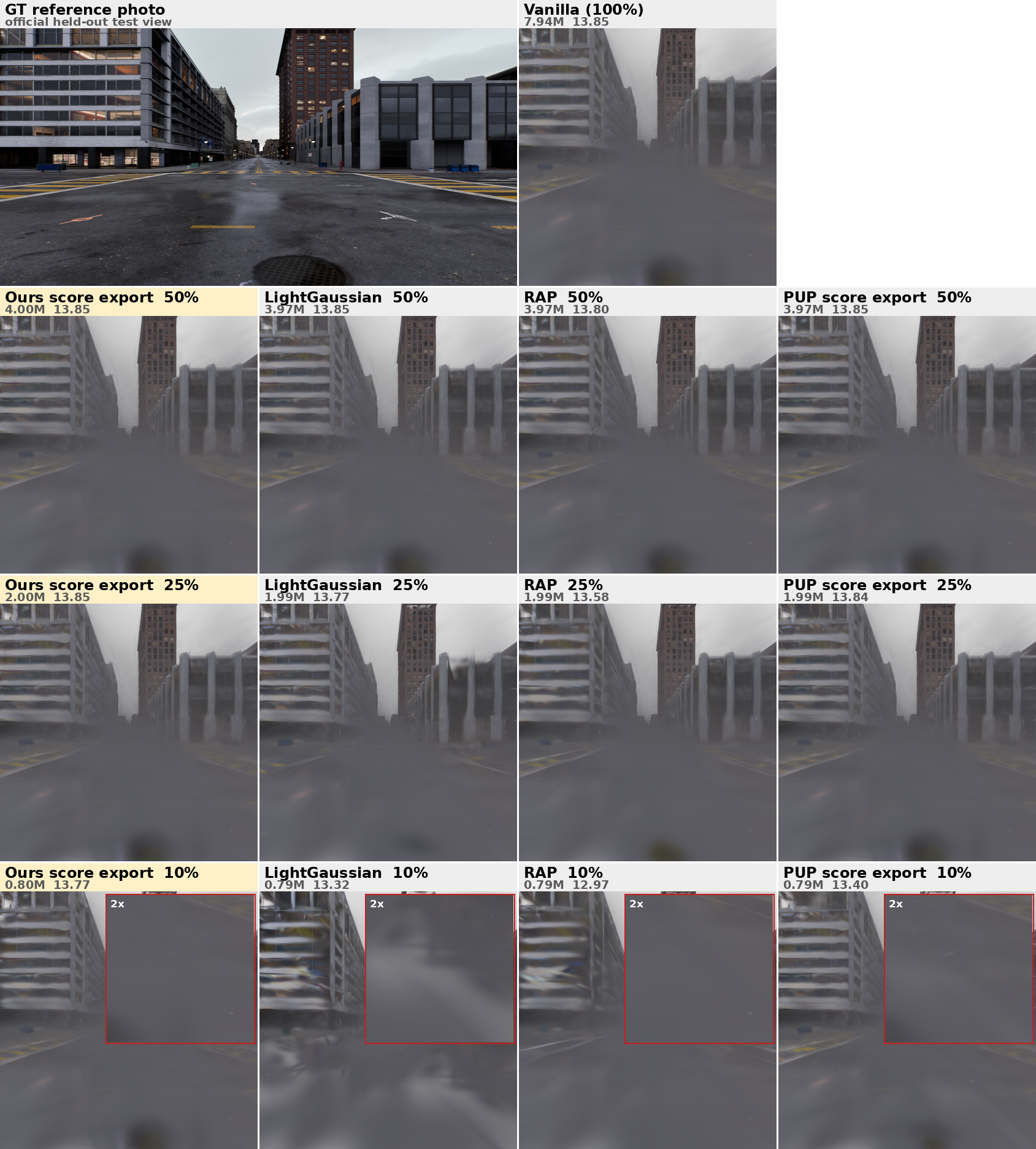}
\caption{Block\_A at matched budgets (official held-out view, PSNR
and primitives below each title): the GT reference and the Vanilla
anchor on top; below, the three budget rows of
Table~\ref{tab:blocka}: our score export (highlighted) against
LightGaussian, RAP, and the PUP score export (our faithful
diagonal-Fisher reimplementation) at 50\%/25\%/10\% of the primitives.
Note the score export's flat curve: no visible difference from Vanilla
down to 25\%, and at 10\% it keeps street markings legible where
LightGaussian blurs and RAP collapses.}
\label{fig:pl4_visual}
\end{figure*}

\begin{figure*}[t]
\centering
\includegraphics[width=\linewidth]{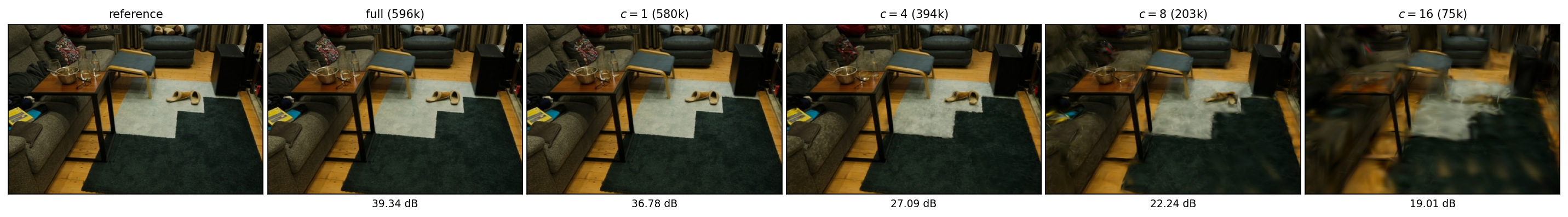}
\caption{The same factoring-tree LOD on an indoor scene (room, Mip-NeRF
360): one held-out view at $c\in\{1,4,8,16\}$ plus the reference
(numbers beneath the tiles: PSNR against the full render on this crop).
Indoor sub-pixel redundancy saturates quickly, so the ladder compresses
less than the city-scale case and visible degradation sets in earlier
(at $c{=}8$ here), consistent with the near-field saturation of
Fig.~\ref{fig:garden_orbit}.}
\label{fig:room_lod_ladder}
\end{figure*}

\begin{figure*}[t]
\centering
\includegraphics[width=\linewidth]{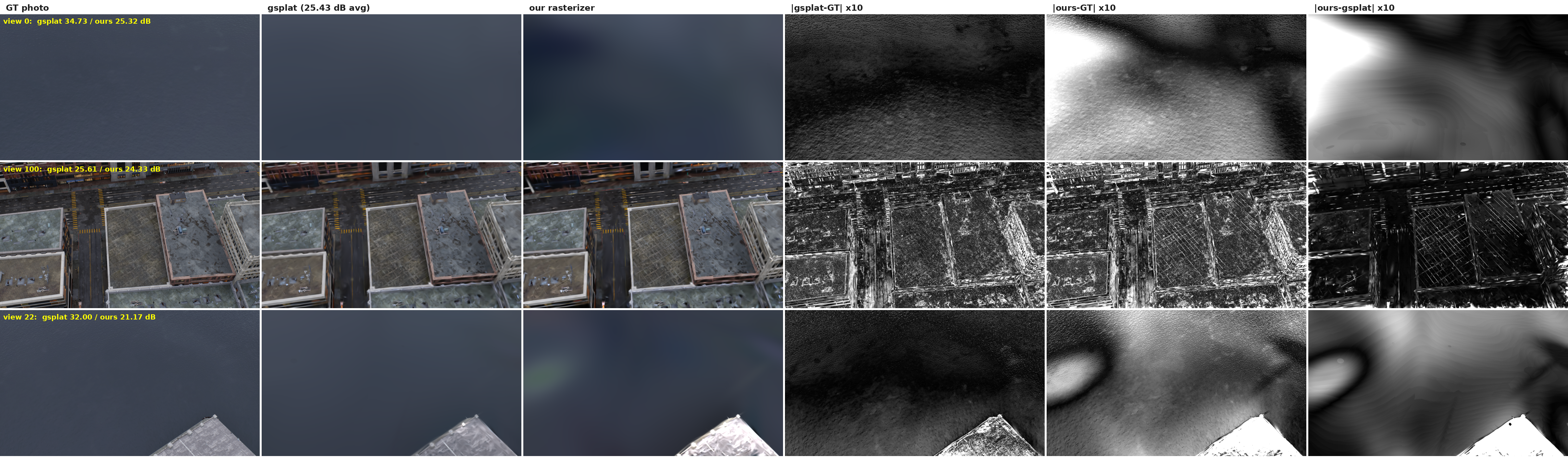}
\caption{Where the cross-renderer gap lives (official 207-view
protocol, v3 model, native $1920\times1080$, full SH). Three views---a
clean aerial view (top), a typical street view (middle), and the
worst-case haze view (bottom): GT photo, the gsplat reference, our CUDA
rasterizer, and absolute-difference maps ($\times 10$). The
ours-vs-gsplat difference concentrates in the low-altitude haze band
(giant low-opacity splats), while building geometry and upper regions
agree; the same concentration appears in both renderers' GT residuals.
This is why every model-level number in this paper is measured with each
side's reference rasterizer, and our rasterizer is reported on the speed
axis (Table~4 of the main paper).}
\label{fig:renderer_gap}
\end{figure*}

\begin{table*}[t]
\centering
\small
\begin{tabular}{lllrrr}
\toprule
scene & budget & schedule & PSNR~dB & ours-prune & published \\
\midrule
bicycle & 2222k & orig & \textbf{23.78} & $23.76$ & 24.77 (C3DGS) \\
bicycle & 2260k & orig & 23.72 & $23.71$ & 25.04 (EAGLES) \\
\midrule
stump & 1732k & orig & 23.75 & $23.81$ & 26.46 (C3DGS) \\
stump & 2220k & orig & \textbf{23.86} & $23.86$ & 26.78 (EAGLES) \\
\midrule
treehill & 2006k & orig & 21.85 & $21.99$ & 22.65 (C3DGS) \\
treehill & 1600k & orig & 21.80 & $21.90$ & 22.69 (EAGLES) \\
\bottomrule
\end{tabular}
\caption{Outdoor extension of Table~1 of the main paper. BCT tracks the
same-ruler prune scan within $\pm0.15$~dB on every point (training at
the budget reproduces the floor family), while the published points sit
above it by $0.8$--$2.9$~dB (stump: $2.7$--$2.9$~dB). The offset is
\emph{not} better use of the budget: it is the redundancy quantified in
Section~6.1 of the main paper (frontier slope $\approx-0.04$).}
\label{tab:bct-ext}
\end{table*}

\begin{table*}[t]
\centering
\small
\begin{tabular}{lrrrrl}
\toprule
Method & \#primitives & PSNR$\uparrow$ & SSIM & LPIPS & vs.\ Vanilla \\
\midrule
Vanilla 3DGS & $\sim$339k & 37.19$\pm$0.02 & 0.9860 & 0.0090 & --- \\
MCMC@171k~\cite{kheradmand2024mcmc} & 171k & 36.62$\pm$0.03 & 0.9853 & 0.0117 & $-0.57$ \\
Taming@94k~\cite{mallick2024taming} & $\sim$92k & 36.11 & --- & --- & $-1.08$ \\
EDC~\cite{deng2024edc} & $\sim$334k & 37.18$\pm$0.01 & --- & --- & $-0.01$ \\
H3DGS~\cite{kerbl2024h3dgs} & $\sim$436k & 37.39$\pm$0.01 & --- & --- & $+0.20$ \\
\textbf{BCT stop-rule (ours)} & $\sim$151k & \textbf{36.85$\pm$0.01} & 0.9845 & \textbf{0.0107} & $-0.34$ \\
\bottomrule
\end{tabular}
\caption{lego, 3 seeds for PSNR; SSIM/LPIPS at matched near-$N$ runs
(per-report protocol). The BCT stop rule (Section~5.1 of the main paper) reaches within $0.34$~dB of Vanilla
quality at 44\% of the primitives and beats MCMC at 12\% fewer primitives;
Taming needs 39\% fewer primitives but loses 0.74~dB. Density-controller
baselines (EDC, H3DGS, our re-runs under the identical protocol) do not
expose a primitive budget: EDC converges to Vanilla-like counts, H3DGS
trades \emph{more} primitives ($+29\%$) for $+0.2$~dB; our rule is the
only method on the budget axis.}
\label{tab:lego}
\end{table*}

\begin{table*}[t]
\centering
\small
\begin{tabular}{lrrrl}
\toprule
Scene & Method & \#primitives & PSNR & \% of vanilla budget \\
\midrule
chair & Vanilla & 491k & 36.39 & --- \\
 & \textbf{stop-rule} & \textbf{177k} & \textbf{36.06} & \textbf{36\%} \\
ship & Vanilla & 259k & 32.05 & --- \\
 & \textbf{stop-rule} & \textbf{124k} & \textbf{31.83} & \textbf{48\%} \\
bonsai & Vanilla & 1258k & 32.64 & --- \\
 & \textbf{stop-rule} & \textbf{825k} & \textbf{32.46$\pm$0.09} & \textbf{66\%} \\
room & Vanilla & 1108k & 32.66 & --- \\
 & \textbf{stop-rule} & \textbf{600k} & \textbf{32.61$\pm$0.04} & \textbf{54\%} \\
\bottomrule
\end{tabular}
\caption{Multi-scene generalization: $-0.06$ to $-0.33$~dB at 36--66\% of
the vanilla budget (34--64\% saved). Synthetic scenes (chair, ship) and
Mip-NeRF 360 indoor scenes (bonsai, room; held-out protocol; 3 seeds for
the stop rule on bonsai/room, mean$\pm$std, primitive counts are seed means).}
\label{tab:multiscene}
\end{table*}

\begin{table*}[t]
\centering
\small
\begin{tabular}{lllll}
\toprule
Scene & original & 50\% score / opacity & 25\% score / opacity \\
\midrule
lego & 36.85 & \textbf{29.68} / 25.87 & \textbf{22.96} / 20.26 \\
bonsai & 32.37 & \textbf{31.82} / 26.48 & \textbf{26.35} / 21.20 \\
room & 32.56 & \textbf{32.42} / 29.71 & \textbf{29.48} / 21.62 \\
\bottomrule
\end{tabular}
\caption{Contribution-scored vs.\ opacity-ordered post-hoc pruning
(PSNR, dB, matched primitive counts at 50\%/25\% of the model;
CityGaussian test protocol, held-out views). The score export is
$+2.7$ to $+7.9$~dB better at every scene and budget.}
\label{tab:prune_multiscene}
\end{table*}

\begin{table*}[t]
\centering
\small
\resizebox{\linewidth}{!}{%
\begin{tabular}{lllllll}
\toprule
system & coupling & model & budget axis & quality @ budget & fps @ $1080$p & VRAM peak \\
\midrule
Octree-GS & training & own (8.01M anchors) & level (5) & $27.23$@$145$k vis.\ / $19.47$@$29$k vis. & $22.9\to5.5$ & 30.4 GiB \\
CG V2 LOD & training & Aerial (8.56M) & block level & n/a post-hoc$^\dagger$ & n/a & n/a \\
\textbf{Factoring LOD} & post-hoc & Aerial (8.56M) & $c$ (continuous) & $28.4\to17.5$ dB & $49.5$--$76.8$ (ladder) & \textbf{4.4 GiB} \\
\bottomrule
\end{tabular}}
\caption{Training-coupled LOD vs.\ our post-hoc layer on MatrixCity.
Octree-GS: our re-render of the official Block\_all checkpoint under a
136-view held-out split (levels 0--4; full model $=27.23$ dB at $145$k
visible primitives; authors report $26.99$ dB under their own split);
per-level CUDA-synchronized averages
$22.95$/$11.19$/$6.37$/$5.50$/$5.47$ fps at levels 0--4, same machine
(RTX 4080 SUPER), native $1920\times1080$. Its ``visible
primitives'' are anchors: the checkpoint stores $8.01$M of them,
each decoding to $10$ neural Gaussians (${\sim}80$M decodable,
$2.5$~GB), while per-view LOD activates only ${\sim}94$k anchors
(${\sim}0.9$M Gaussians); the quality column spans the ends of its
level curve ($27.23$ dB at level 4). Our row: the official-protocol 207-view
trajectory at the full $6.89$M set, $71.0$~fps mean (p5 $55.2$) on
our CUDA rasterizer ($54.3$~fps on the reference gsplat path),
full SH;
the capacity ladder spans $49.5$--$76.8$~fps at $1080p$
(Table~4 of the main paper). $^\dagger$CityGaussianV2's runtime LOD requires
partitioned training artifacts (partitions.pt plus one trained
checkpoint per partition block) that the officially released single
Aerial checkpoint does not include
(Section~6.4 of the main paper); our layer consumes that same checkpoint
directly and spans the measured $c$ curve of Table~3 of the main paper.
VRAM peak (nvidia-smi sampled at 1600$\times$900, same
Block\_all views, each side's official path, RTX 4080 SUPER 32\,GB):
Octree-GS transient peak 30.4\,GiB (steady-state $\approx$15.9\,GiB) vs.\
ours 4.4\,GiB, $6.9\times$ lower at peak. Ours is the gsplat reference
path with full SH.}
\label{tab:e1}
\end{table*}

\begin{table*}[t]
\centering
\small
\begin{tabular}{llrrrr}
\toprule
budget & method & \#primitives & PSNR & SSIM & LPIPS$\downarrow$ \\
\midrule
100\% & Vanilla & 7.94M & 13.85 & 0.519 & 0.553 \\
\midrule
50\% & LightGaussian & 3.97M & \textbf{13.85} & \textbf{0.520} & \textbf{0.552} \\
50\% & RAP & 3.97M & 13.80 & 0.519 & 0.558 \\
50\% & PUP score export & 3.97M & 13.85 & 0.519 & 0.553 \\
50\% & score export (ours) & 4.00M & \textbf{13.85} & 0.519 & 0.553 \\
50\% & $\theta$+opacity (ours) & 4.00M & 13.44 & 0.513 & 0.530 \\
\midrule
25\% & LightGaussian & 1.99M & 13.77 & 0.518 & \textbf{0.550} \\
25\% & RAP & 1.99M & 13.58 & 0.514 & 0.573 \\
25\% & PUP score export & 1.99M & 13.84 & 0.519 & 0.553 \\
25\% & score export (ours) & 2.00M & \textbf{13.85} & \textbf{0.519} & 0.553 \\
25\% & $\theta$+opacity (ours) & 2.00M & 11.14 & 0.470 & \textbf{0.518} \\
\midrule
10\% & RAP & 0.79M & 12.97 & 0.503 & 0.601 \\
10\% & PUP score export & 0.79M & 13.40 & 0.509 & 0.557 \\
10\% & LightGaussian & 0.79M & 13.32 & 0.511 & 0.563 \\
10\% & score export (ours) & 0.80M & \textbf{13.77} & \textbf{0.515} & 0.557 \\
10\% & $\theta$+opacity (ours) & 0.80M & 8.63 & 0.405 & \textbf{0.547} \\
\bottomrule
\end{tabular}
\caption{Block\_A pruning at matched budgets (gsplat, DC-only,
$1000^2$, 5 held-out views vs.\ real photos). Absolute values are low
because evaluation is DC-only against photographs; only relative
comparisons are meaningful. The score export (Section~5.3 of the main paper)
is Vanilla-identical at $50\%$/$25\%$ and dominates PSNR/SSIM at
$10\%$; deeper tiers: $13.49$ ($5\%$), $12.97$ ($2.5\%$), $12.53$
($1.3\%$). $\theta$+opacity leads LPIPS at extreme budgets. PUP row: our
single-pass diagonal-Fisher reimplementation of the PUP~3D-GS score
(their full pipeline adds multi-round finetuning).}
\label{tab:blocka}
\end{table*}

\begin{figure*}[t]
\centering
\includegraphics[width=\linewidth]{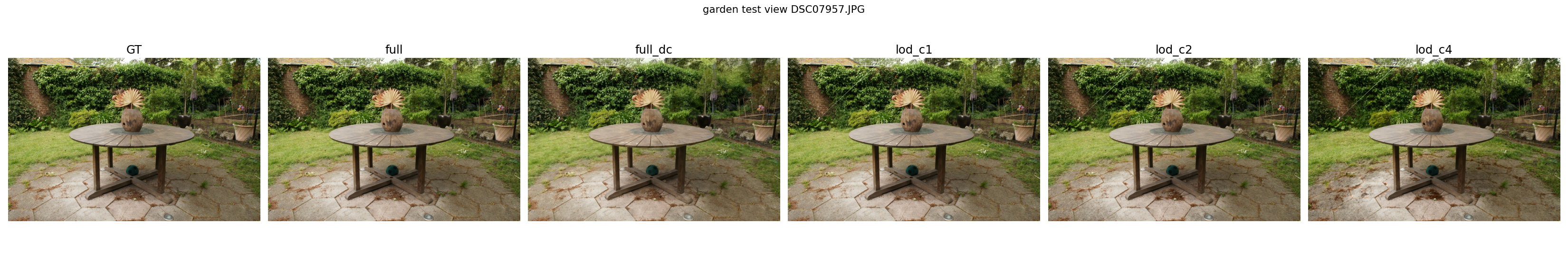}
\caption{garden, held-out trajectory view: reference photo, full render,
full render with DC-only color, and the LOD tiers of Table~2 of the main paper. DC-only aggregation costs $\sim$5 dB
(Section~6.7 of the main paper) while the first-moment SH tiers stay
close to the full render.}
\label{fig:garden_lod}
\end{figure*}

\begin{figure*}[t]
\centering
\includegraphics[width=\linewidth]{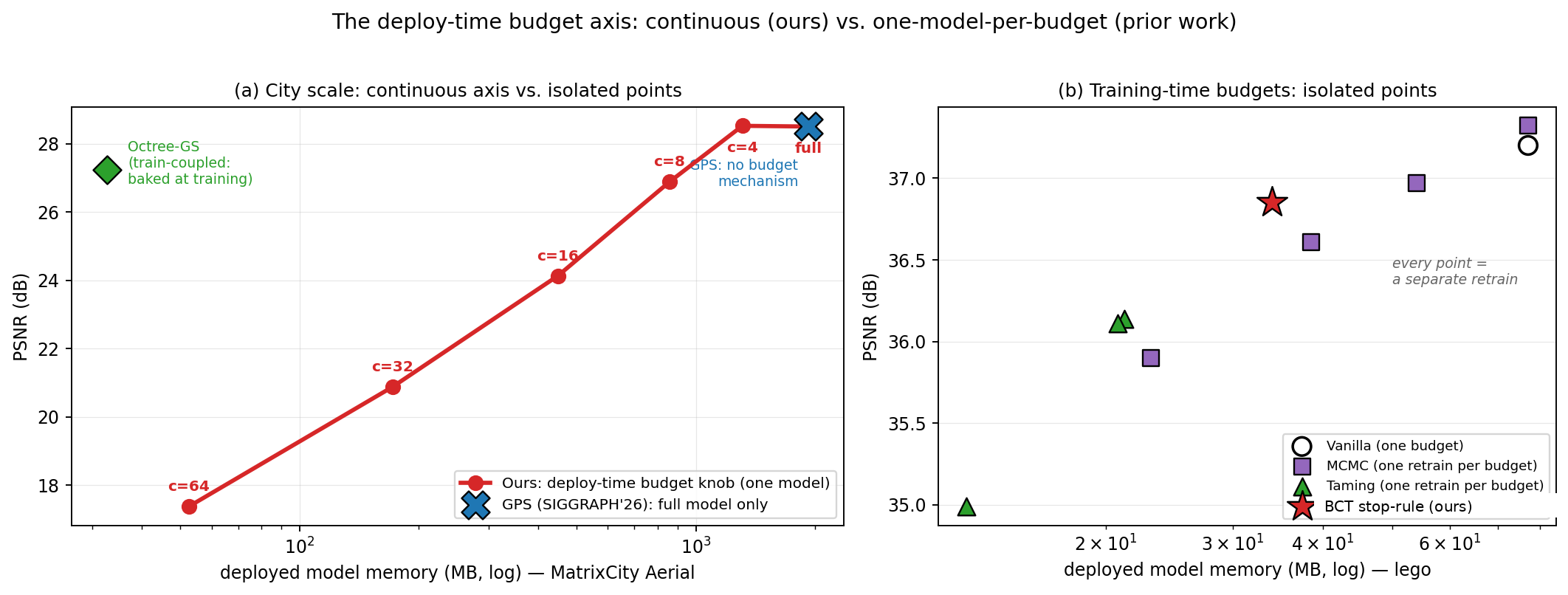}
\caption{The deploy-time budget axis: continuous (ours) vs.\
one-model-per-budget (prior work). (a)~MatrixCity Aerial: our single
model sweeps a continuous memory--quality curve; GPS renders the full
model only, and Octree-GS's LOD is baked at training time.
(b)~Training-time budgets (lego): MCMC and Taming produce one fixed-budget
model per retrain; the BCT stop rule reaches MCMC-level quality with one training run
at a smaller budget.}
\label{fig:sota_convergence}
\end{figure*}

\begin{figure*}[t]
\centering
\includegraphics[width=0.98\linewidth]{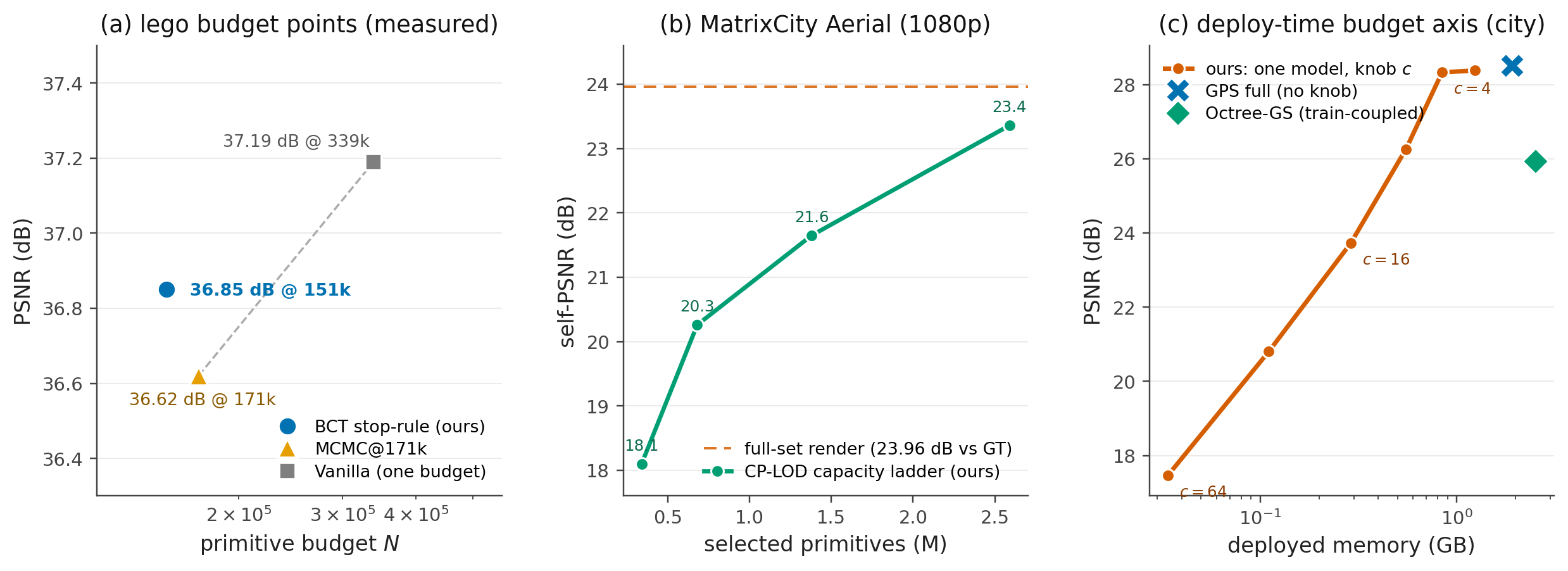}
\caption{Quality--budget behaviour across three axes. (a) training-time
primitive budgets on lego (measured points,
Table~\ref{tab:lego}): at 151k the BCT stop rule beats MCMC's own
budget point (171k) with $12\%$ fewer primitives; Vanilla is a single
point at 339k. (b) deploy-time capacity ladder on MatrixCity
Aerial (1080p): self-PSNR rises smoothly with the selected set toward the
full-model reference (dashed, vs.\ GT on this view). (c) the deploy-time
budget parameter (city scale, gsplat model sizes): one
trained model traces a continuous quality--memory curve; GPS offers only
its full model and Octree-GS is baked at training time.}
\label{fig:budget-axes}
\end{figure*}

\subsection{Memory--quality Pareto vs.\ GPS on MatrixCity}
\label{app:m4-pareto}

\begin{figure}[h]
\centering
\includegraphics[width=0.85\linewidth]{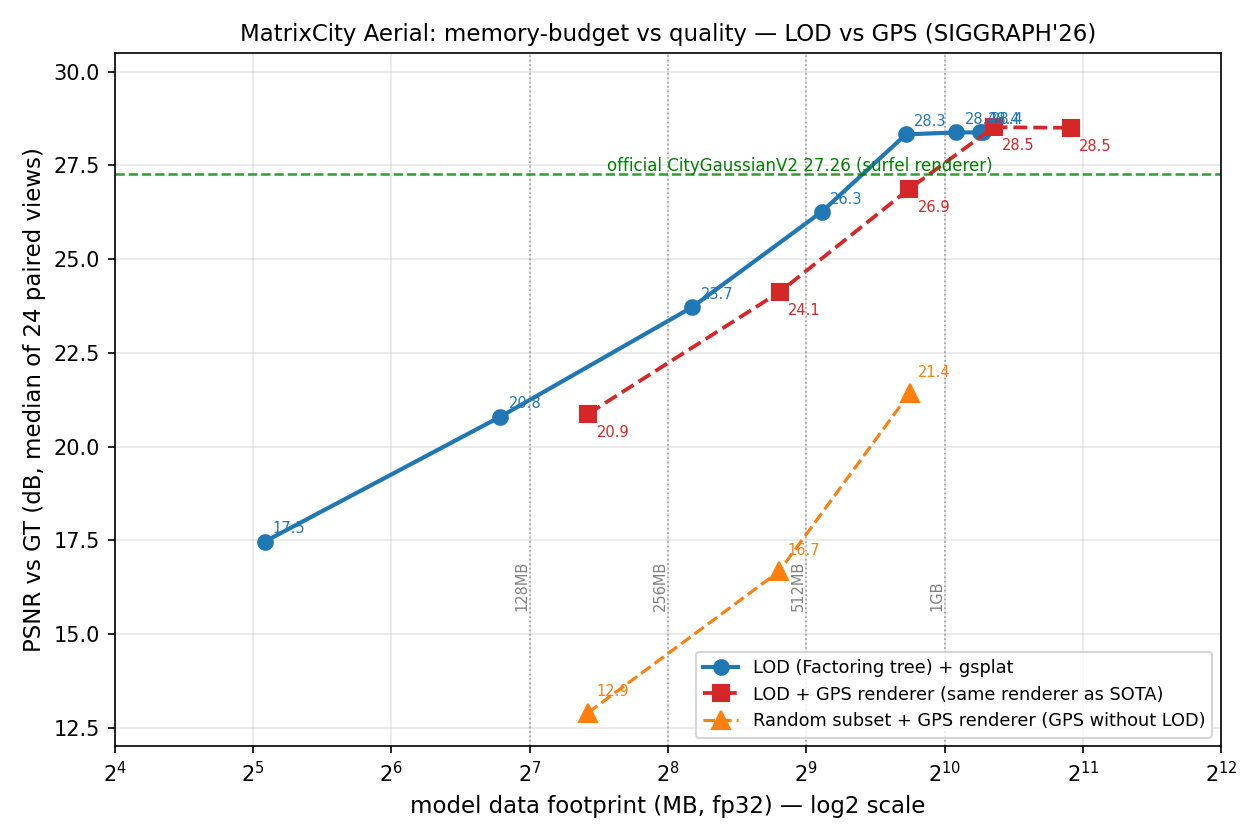}
\caption{Memory--quality Pareto on MatrixCity Aerial: our LOD curve, the same LOD exports through the GPS renderer, and GPS random subsampling.}
\label{fig:pareto_mc}
\end{figure}

Fig.~\ref{fig:pareto_mc}, Fig.~\ref{fig:sota_convergence}, and
Table~3 of the main paper establish the
central claim. \textbf{(a) A continuous budget--quality curve exists:}
1240 MB $\to$ 34 MB at 28.4 $\to$ 17.5 dB, near-lossless at $c\le4$ ($-0.05$~dB).
\textbf{(b) Renderer-agnostic:} exporting each LOD tier to the GPS renderer
reproduces the same curve within 0.5 dB, and $c{=}4$ is lossless under
\emph{both} renderers; the LOD layer carries the quality, not the
rasterizer. \textbf{(c) Same-budget dominance over random
subsampling:} at 44.6\%/23.2\%/8.9\% primitive budgets the LOD selection beats
random subsampling by $+5.5/+7.5/+8.0$~dB (median); at 8.9\% random
sampling is unusable (12.9 dB) while LOD keeps 20.9 dB. \textbf{(d) Memory:}
matching the $\sim$20.8 dB operating point needs 110.6 MB (ours) vs.\
1926 MB (GPS full), $17\times$ savings. This comparison holds everything
constant except the budget mechanism: same trained model, same renderer
(GPS), same cameras, same GT, and the LOD tier is re-rendered through GPS
itself (Section~5 of the main paper); the only difference is that our
axis lets the deployment pick 110.6 MB where GPS offers only the full
model (visual comparison in Fig.~\ref{fig:mc_visual}). Model data is the
fp32 splat bytes of the exported per-budget PLY (152 B/prim 3DGS,
236 B/prim GPS sh3; GPS full process peak 2364 MiB; byte layout
verified against an exported $c{=}8$ PLY).

\begin{table*}[t]
\centering\small
\caption{Floor curves across the 13 public scenes: full-model PSNR, full-range
chord slope (registered negative), steepest three-band window slope
(brackets $-1/2$), floor intactness (all pass).}
\label{tab:eu4scenes}
\begin{tabular}{lrrl}
\toprule
scene & $R_{\mathrm{full}}$ (dB) & chord & steepest \\
\midrule
garden   & 26.54 & $-0.155$ & $-0.26$ \hfill\proved \\
room     & 31.59 & $-0.258$ & $-0.47$ \hfill\proved \\
bonsai   & 32.37 & $-0.275$ & $-0.49$ \hfill\proved \\
bicycle  & 23.75 & $-0.124$ & $-0.20$ \hfill\proved \\
counter  & 29.27 & $-0.231$ & $-0.37$ \hfill\proved \\
flowers  & 20.13 & $-0.101$ & $-0.15$ \hfill\proved \\
kitchen  & 31.02 & $-0.235$ & $-0.41$ \hfill\proved \\
stump    & 23.85 & $-0.088$ & $-0.17$ \hfill\proved \\
treehill & 21.87 & $-0.073$ & $-0.12$ \hfill\proved \\
train    & 22.27 & $-0.141$ & $-0.18$ \hfill\proved \\
truck    & 25.22 & $-0.190$ & $-0.29$ \hfill\proved \\
drjohnson& 28.50 & $-0.203$ & $-0.31$ \hfill\proved \\
playroom & 29.19 & $-0.229$ & $-0.46$ \hfill\proved \\
\bottomrule
\end{tabular}
\end{table*}

\begin{table*}[t]
\centering
\small
\begin{tabular}{lrrr}
\toprule
step & primitives & time & peak GPU \\
\midrule
tree construction & 8.56M & 1.5 s & 19.4 GiB \\
LOD select + render (per view) & $0.1$--$8.6$M & 6--172 ms & 3.3--4.5 GiB \\
full render (gsplat, per view) & 8.56M & 3--13 ms & --- \\
\bottomrule
\end{tabular}
\caption{Offline evaluation protocol: MatrixCity on an A100 GPU (GPU 2,
80 GB), Python scheduler + gsplat. Tree construction is once per model;
per-view selection in Python dominates offline LOD latency. The real-time
CUDA path removes exactly this cost: $0.8$\,ms per-frame preprocess
(selection$+$gather) included in the $71.0$~fps full-set $1080p$ official
trajectory of our CUDA rasterizer (Table~4 of the main paper,
Section~6.6 of the main paper).}
\label{tab:eff}
\end{table*}

\end{document}